\documentclass{article}
\usepackage{amsmath, amssymb}
\usepackage{amsfonts}
\usepackage{appendix}
\usepackage{graphicx}
\usepackage{color}
\usepackage{url}
\usepackage{bm}
\usepackage[all]{xy}
\usepackage{arydshln}
\usepackage{colortbl}
\usepackage{tabularx}
\usepackage{makecell}
\usepackage{xcolor}
\usepackage{threeparttable}

\newtheorem{fact}{Fact}

\newtheorem{theorem}{Theorem}
\newtheorem{lemma}{Lemma}
\newtheorem{corollary}{Corollary}

\newtheorem{definition}{Definition}
\newtheorem{proposition}{Proposition}

\newenvironment{proof}{{\noindent\it Proof}\quad}{\hfill $\square$\par}
\newtheorem{example}{Example}

\newcommand{\Z}{\ensuremath{\mathbb Z}}

\newcommand{\C}{\ensuremath{\mathbb C}}

\newcommand{\ls}[1]
{\dimen0=\fontdimen6\the\font\lineskip=#1\dimen0
	\advance\lineskip.5\fontdimen5\the\font
	\advance\lineskip-\dimen0
	\lineskiplimit=0.9\lineskip
	\baselineskip=\lineskip
	\advance\baselineskip\dimen0
	\normallineskip\lineskip\normallineskiplimit\lineskiplimit
	\normalbaselineskip\baselineskip
	\ignorespaces}

\begin{document}
	
	\bibliographystyle{abbrv}
	
	\title{New Constructions of Complementary Sequence Pairs over $4^q$-QAM
}
\footnotetext{The material in this paper was  present in part at the IEEE International Symposium on Information Theory, Los Angeles, California, USA June, 2020.}
\author{Zilong Wang$^1$, Erzhong Xue$^{1}$, Guang Gong$^2$, IEEE Fellow\\
	\small $^1$ State Key Laboratory of Integrated Service Networks, Xidian University \\[-0.8ex]
	\small Xi'an, 710071, China\\
	\small $^2$Department of Electrical and Computer Engineering, University of Waterloo \\
	\small Waterloo, Ontario N2L 3G1, Canada  \\
	\small\tt zlwang@xidian.edu.cn, 2524384374@qq.com, ggong@uwaterloo.ca\\
}
	\maketitle

	\thispagestyle{plain} \setcounter{page}{1}
	\begin{abstract}

The previous constructions of quadrature amplitude modulation (QAM) Golay complementary sequences (GCSs)  were generalized as $4^q $-QAM GCSs of length $2^{m}$  by Li \textsl{et al.} (the generalized cases I-III for $q\ge 2$) in 2010 and Liu \textsl{et al.} (the generalized cases IV-V for $q\ge 3$) in 2013 respectively. Those sequences are presented as  the combination of the
quaternary standard GCSs and compatible offsets. By providing new compatible offsets based on the factorization of the integer $q$,  we proposed two new constructions of $4^q $-QAM GCSs, which have the  generalized cases I-V as special cases. The numbers of the proposed GCSs (including the generalized cases IV-V) are  equal to the product of the number of the quaternary standard GCSs  and the number of the compatible offsets. For $q=q_{1}\times q_{2}\times \dots\times q_{t}$ ($q_k>1$), the number of new offsets in our first construction is lower bounded by a polynomial of $m$ with degree $t$, while the numbers of offsets in the generalized cases I-III and IV-V are a linear  polynomial of $m$ and a  quadratic  polynomial of $m$, respectively. In particular, the numbers of  new  offsets in our first construction is seven times more than that in the generalized  cases IV-V for $q=4$. We also show that the numbers of  new  offsets in our two constructions is lower bounded by a cubic polynomial of $m$ for $q=6$.
Moreover,  our proof implies that all the mentioned  GCSs over QAM in this paper can be regarded as projections of Golay complementary arrays of size $2\times2\times\cdots\times2$.
\end{abstract}

	
\ls{1.5}
	\section{Introduction}

A pair of sequences is called a Golay complementary pair (GCP) \cite{Golay1961Complementary} if their aperiodic autocorrelation sums for any
nonzero shifts are all equal to zero. Each sequence in the GCP
is called a Golay complementary sequence (GCS). The concept of binary GCP was extended later to the polyphase case \cite{Sivaswamy78} and
complementary sequence sets \cite{Tseng1972Complementary}. These sequences have found numerous applications in various fields of science and
engineering, especially in orthogonal frequency-division multiplexing (OFDM) systems. One of the major impediments to deploying OFDM is the high peak-to-mean envelope power ratio (PMEPR) of uncoded OFDM signals.
PMEPR reduction in OFDM transmission can be implemented  by using codes constructed from the sequences in the complementary sequence sets, especially GCSs \cite{Boyd1986Multitone,Popovic1991Synthesis}.

GCPs were initially constructed by the recursive methods \cite{Golay1961Complementary,Budi1990}.
An extensive study on this topic was made by Davis and
Jedwab in \cite{Davis1999Peak} by a direct  construction of polyphase
GCSs based on generalized Boolean functions (GBFs), which have been referred to as the standard GCSs  subsequently.  Non-standard GCPs were studied in \cite{Li05,Fiedler06,Fiedler2008A} and complementary sequence sets were constructed in \cite{Paterson2000Generalized,Schmidt2006Complementary} based on  GBFs later on.

All the aforementioned sequences are constructed over the phase-shift keying (PSK) constellations. Since quadrature amplitude modulation (QAM) are widely employed in high rate OFDM transmissions, $16$-QAM sequences based on weighted quaternary PSK (QPSK) GCSs were studied by R{\"o}\ss ing and Tarokh \cite{Robing2001A} in 2001. Chong \textsl{et al.} \cite{Chong2003A} then proposed a construction of 16-QAM GCSs based on  standard GCSs over QPSK and first-order offsets of Reed-Muller codes in three cases. It was pointed out in \cite{Chong2003A} that an OFDM system with 16-QAM GCSs  has a higher code rate than that with  binary or quaternary standard GCSs, given the same PMEPR constraint. In 2006, Lee and Golomb \cite{Lee2006A} proposed a construction of $64$-QAM GCSs with the weighted-sum of three standard GCSs over QPSK and first-order offsets in five cases. Further improvements on the constructions of GCSs over 16-QAM and 64-QAM were given by Li \textsl{et al.}  \cite{Li2008Comments} and Chang \textsl{et al.} \cite{Chang2010New} later on. These results
were extended to the general construction GCSs  over $4^q$-QAM  by Li \textsl{et al.} \cite{Li2010A} in 2010 and Liu \textsl{et al.} \cite{Liu2013New} in 2013, respectively. All these GCSs over QAM are constructed based on standard QPSK GCSs and compatible offsets. Depending on the algebraic structure of the compatible offsets, the GCSs proposed in \cite{Li2010A,Liu2013New} are referred to as the generalized  cases I-III and the generalized cases IV-V, respectively.

In 2018, Budi{\v{s}}in and Spasojevi$\acute{c}$ \cite{Budi2018PU} introduced a new  recursive  algorithm in multiplicative form to generate GCPs over QAM by para-unitary (PU) matrices, where any element of a sequence can be generated by indexing the entries of unitary matrices with the binary representation of the discrete time index. Sequences derived from $M$ unitary matrices over QAM constellation are referred to as the $M$-Qum case. It is shown that the 1-Qum case and 2-Qum case can generate the sequences in the generalized cases I-III \cite{Li2010A} and cases IV-V \cite{Liu2013New}, respectively. Moreover, a large number of new GCSs over QAM
are produced from the $M$-Qum case when $M\geq 2$. However, for given $q$ and sequences length $2^m$, the acceptable unitary matrices over QAM can only be obtained by exhaustive search, and the lack of explicit algebraic expression of GCSs leads to unexpected  duplication for $M$-Qum case when $M\geq 2$.

In this paper, we propose two new constructions of GCSs and GSPs over $4^q$-QAM of length $2^{m}$.
For the generalized cases I-III \cite{Li2010A}, the generalized cases IV-V \cite{Liu2013New}, and the new constructions in this paper, the GCSs  are all expressed by the combination of standard GCSs over QPSK and compatible offsets. We re-express the compatible offsets in the generalized cases I-III and IV-V by the so-called
$\vec{d}$-vectors and $\vec{b}$-vectors derived from the Set $\mathcal{C}$ (in Definition \ref{Set-C}) and  non-symmetrical Gaussian integer pair, respectively. Furthermore, 
the new compatible offsets in this paper are constructed by the $\vec{b}$-vectors and the new $\vec{d}$-vectors based on the factorization of the integer $q$. If $q$ is a prime, the proposed constructions in this paper coincide with the generalized cases I-V. If $q$ is a  composite number, the proposed constructions  comprise of not only  the generalized cases I-V, but also a great many new GCSs.

The numbers of the proposed GCSs including the generalized cases I-V are all equal to the product of the number of the quaternary standard GCSs and the number of the compatible offsets. Since the quaternary standard GCSs was given in  \cite{Davis1999Peak}, the numbers of the proposed GCSs are determined by the number of the compatible offsets. 
It was shown in \cite{Li2010A} and \cite{Liu2013New} that the numbers of offsets in the generalized cases I-III and IV-V are a linear  polynomial of $m$ and a  quadratic  polynomial of $m$, respectively. We show that, for $q=q_{1}\times q_{2}\times \dots\times q_{t}$ ($q_k>1$), the number of new offsets in our first construction is lower bounded by a polynomial of $m$ with degree $t$.  In particular, for $q=4$, the numbers of  new  offsets in our first construction is seven times more than that in the generalized  cases IV-V. We also show that the numbers of  new  offsets in our two constructions is lower bounded by a cubic polynomial of $m$ for $q=6$.

Although the GCSs over QAM proposed in our constructions are  represented by weighted-sum of standard GCSs over QPSK, the methodology in this paper is totally different from the aforementioned references which also constructed GCSs over QAM by weighted-sum of standard GCSs. Our ideas here are inspired by the PU algorithm  \cite{Budi2018PU} and the Golay array pairs (GAPs) \cite{Array2}, and benefit from a recent proposed approach to extract GBFs from PU matrices \cite{CCA}.  GAPs and their relationship with GCPs over PSK by a three-stage construction process were introduced by F. Fiedler \textsl{et al.} \cite{Array2} in 2008. We extend this idea from PSK  modulation to  QAM modulation, and propose a mapping from a GAP of size $\underbrace{2\times 2 \times \cdots \times 2}_m$ to a large number of GCPs over QAM of length $2^m$.
 We also make a connection of the construction of  GAPs  and some specified PU matrices with multi-variables over $4^q$-QAM.  
Finally, we propose two constructions of these PU matrices, and extract the corresponding GAPs over QAM, from which we obtain our new constructions of GCPs and GCSs over QAM.

The rest of this paper is organized as follows. In the next section, we introduce the definitions of GCP and GCS, and revisit the known constructions of the GCPs over QAM. In Section 3, we present two new constructions of  $4^{q}$-QAM GCSs  including the generalized cases I-V as special cases in Theorems 1 and 2. Enumerations of the new GCSs other than the generalized cases I-V are given in Section 4. The proofs of our main results are presented in the succeeding three sections.  In Section 5,  we show our viewpoint to construct GCPs over $4^q$-QAM by GAPs and PU matrices. We then derive the PU form (in Theorems 5 and 6) and the array form (in Theorem 8) of our result in Section 6 and Section 7, respectively.  We conclude the paper in Section 8.

\section{Preliminaries}
The following notations will be used throughout the paper.
\begin{itemize}
	\item $q, p, m, L$ are all positive integers, where $0\le p<q $.
	\item 
	$\Z_4$ is the residue class ring modulo $4$. $\xi=\sqrt{-1}$ is a fourth primitive root of unity.
	\item $\mathbb{F}_2$ is the finite field with two elements, and $\mathbb{F}_2^m$ is $m$-dimensional vector space over $\mathbb{F}_2$.
    \item $\C$ is the complex field. For any $\alpha\in\C$, $ \overline{\alpha} $ is the conjugation of $ \alpha $.
	\item $\bm{x}=(x_1, x_2, \cdots x_{m})\in\mathbb{F}_2^m$, where each $x_i$ are Boolean variables for $1\leq i\leq m$.
.
	\item $\pi$ is a permutation of symbols $\{1, 2, \cdots, m\}$.
\item $x_{\pi(0)}$ and $x_{\pi(m+1)}$ as \lq fake\rq \ variables which always equal to $0$



\end{itemize}

\subsection{Golay Complementary Pairs}
Let $F(y)=(F(0), F(1),\cdots, F(L-1))$ be a complex-valued sequence of length $L$.
The {\em aperiodic auto-correlation} of  $F(y)$ at shift ${\tau}$ ($1-L\le\tau\le L-1$) is defined by
	$$C_{{F}}({\tau})=
	\sum_{{y}}{F({y}+{\tau})\cdot \overline{F}({y})},
	$$
	where ${F({y}+{\tau})\cdot \overline{F}({y})}=0$ if  ${F({y}+{\tau})}$ or ${F({y})}$ is not defined.

	A pair of sequences $ \{{F}(y), {G}(y)\}$ is said to be a {\em Golay  complementary  pair} (GCP) if
	\begin{equation}
	{C}_{{F}}(\tau)+{C}_{{G}}(\tau)=0,\quad(\forall \tau\ne 0).\label{equation_GCP}
	\end{equation}
	And  either sequence in a GCP is called a  {\em Golay complementary sequence} (GCS) \cite{Golay1961Complementary}.

\subsection{GCPs over QPSK}

A {\em generalized Boolean function} (GBF)  $f(\bm{x})$ (or $ f( x_1, x_2, \cdots, x_{m})) $ over $\Z_{4}$ is a function from $\mathbb{F}_2^m$ to $\mathbb{Z}_{4}$.
Such a function can be uniquely expressed as a linear combination over $\mathbb{Z}_{4}$ of the monomials \[1, x_1, x_2, \cdots, x_{m}, x_1x_2, x_1x_3, \cdots, x_{m-1}x_{m}, \cdots, x_1x_2x_3\cdots x_{m},\]
where the coefficient of each monomial belongs to $\mathbb{Z}_{4}$.

For $0\le y< 2^m$, $y$ can be written uniquely in a binary expansion as
$y=\sum_{j=1}^{m}x_j\cdot 2^{j-1}$
where $x_j\in\{0, 1\}$.
Then a sequence ${F(y)}$  of length $L=2^m$ over QPSK  can be associated with a GBF  $f(\bm{x})$ over $\Z_{4} $ by
\begin{equation}
F(y)=\xi^{f(\bm{x})}.
\end{equation}

There are several constructions of GCPs over QPSK based on GBFs, such as
 \cite{Davis1999Peak,Li05,Fiedler06,Fiedler2008A}.
The most typical GCPs are called standard GCPs \cite{Davis1999Peak},
associated with GBFs over $\Z_4$ given below.
\begin{fact}[\cite{Davis1999Peak}]
For  GBF
	\begin{equation}
	{f}(\bm{x})={2}\cdot\sum_{j=1}^{m-1}x_{\pi(j)}x_{\pi(j+1)}+\sum_{j=1}^{m}c_{j}\cdot x_{j}+c_0,\label{equation_f(x)}
	\end{equation}	
	where $c', c_{j}\in \Z_4 (0\le{j}\le m)$, the  sequence pair associated with the GBFs over $\Z_4$
		\begin{equation*}
		\left\{
		\begin{aligned}
		&{f}(\bm{x}),\\
		&{f}(\bm{x})+2x_{\pi(1)}+c',
		\end{aligned}\right.
		\quad\text{or}\quad
		\left\{
		\begin{aligned}
		&{f}(\bm{x}),\\
		&{f}(\bm{x})+2x_{\pi(m)}+c',
		\end{aligned}\right.
		\end{equation*}
$c'\in\Z_4$ form a GCP over QPSK.	
\end{fact}

\subsection{GCPs over QAM}

In this paper, a {\em vectorial GBF (V-GBF)} is a function from $\mathbb{F}_2^m$ to $\mathbb{Z}_{4}^{q}$, denoted by \[\vec{f}(\bm{x})=(f^{(0)}(\bm{x}),f^{(1)}(\bm{x}),\cdots,f^{(q-1)}(\bm{x}))^{T}, \]
where each component function $f^{(p)}(\bm{x}) (0\le p<q )$ is a GBF over $\Z_{4}$. Let $\vec{1}$ denote the $q$-dimensional vector $(1,1,\cdots,1)^{T}$.
In this subsection, we revisit the constructions of GCPs over QAM by V-GBFs.

A sequence over $4^q$-QAM  can be viewed as the weighted sums of $q$ sequences over QPSK, with weights in the ratio of $2^{q-1}:2^{q-2}:\dots:1$ \cite{Li2010A} (by ignoring the factor $e^{\xi\pi/2}$).
Then a sequence over $4^q$-QAM of length $2^m$  can be associated with a V-GBF
 $\vec{f}(\bm{x})=(f^{(0)}(\bm{x}),f^{(1)}(\bm{x}),\cdots,f^{(q-1)}(\bm{x}))$ over $\Z_{4} $ by
\begin{equation}\label{equation_QAMBoolean}
F(y)=\sum_{p=0}^{q-1}2^{q-1-p}\cdot  \xi^{f^{(p)}(y)},
\end{equation}
where
$y=\sum_{j=1}^{m}x_{j}\cdot 2^{j-1}$ and ${f^{(p)}(y)}=f^{(p)}(\bm{x})\; (0\le p<q)$.
Obviously,  the sequences over QPSK can be seen as a special case of QAM sequences when $q=1$.

The GCPs $\{{F}(y) , {G}(y)\}$ of length $2^m$  over $4^{q}$-QAM were well studied in the literature. Their associated  V-GBFs $\{\vec{f}(\bm{x}),\vec{g}(\bm{x})\}$ are usually presented by
\begin{equation}\label{equation_QAMGolaypair}
	\left\{
	\begin{aligned}
	&\vec{f}(\bm{x})={f}(\bm{x})\cdot\vec{1}+\vec{s}(\bm{x}),\\
	&\vec{g}(\bm{x})=\vec{f}(\bm{x})+\vec{\mu}(\bm{x}),
	\end{aligned}\right.
	\end{equation}
where ${f}(\bm{x})$ are standard GCSs in form (\ref{equation_f(x)}), $\vec{s}(\bm{x})=({s}^{(0)}(\bm{x})=0,{s}^{(1)}(\bm{x}),\cdots,{s}^{(q-1)}(\bm{x}))^{T}$ and $\vec{\mu}(\bm{x})=({\mu}^{(0)}(\bm{x}),{\mu}^{(1)}(\bm{x}),\cdots,{\mu}^{(q-1)}(\bm{x}))^{T}$ are called {\em offset} V-GBFs and {\em pairing difference} V-GBFs, respectively.
The generalized cases I-III  \cite{Li2010A} and the generalized cases IV-V \cite{Liu2013New} are re-expressed by the offset V-GBFs $\vec{s}(\bm{x})$ and pairing difference V-GBFs $\vec{\mu}(\bm{x})$ in the following Facts \ref{fact2} and \ref{fact3} respectively.

\begin{definition}[Set $\mathcal{C}$]\label{Set-C}
Let $\mathcal{C}$ be a set consisting of all the vectors $\underline{d}=(d_{0}, d_{1}, d_{2})\in \Z_4^3$ such that $2d_0+d_1+d_2=0$ over $\Z_4$. Then $\underline{d}\in \mathcal{C}$ has
following $16$ possible values:
\begin{eqnarray*}
		&(0,0,0),(0,1,3),(0,2,2),(0,3,1),(1,0,2),(1,1,1),(1,2,0),(1,3,3),\\
		&(2,0,0),(2,1,3),(2,2,2),(2,3,1),(3,0,2),(3,1,1),(3,2,0),(3,3,3).
	\end{eqnarray*}
\end{definition}

Recall $x_{\pi(0)}=x_{\pi(m+1)}=0$ as \lq fake\rq\  variables. The  offset V-GBFs  $\vec{s}(\bm{x})$ in the generalized cases I-III can be re-expressed in a unified form by the following $\vec{d}$-vectors derived from the set $\mathcal{C}$.

\begin{fact}[Generalized cases I-III \cite{Li2010A}]\label{fact2}
For $1\le p\le q-1$ and ${\underline{d}^{(p)}}=({d_{0}^{(p)}}, {d_{1}^{(p)}}, {d_{2}^{(p)}})\in \mathcal{C}$, define the $\vec{d}$-vectors by {$\vec{d_{i}}=(d_{i}^{(0)},d_{i}^{(1)},\dots,d_{i}^{(q-1)})^{T}$}  ($i=0,1,2$). For $0\le \omega\le m$,
$\{{F}(y) , {G}(y)\}$  forms a $4^{q}$-QAM GCP of length $2^m$ if the offset V-GBFs
$$\vec{s}(\bm{x})=\vec{d}_{0}+\vec{d}_{1}\cdot x_{\pi(\omega)}+\vec{d}_{2}\cdot x_{\pi(\omega+1)},$$
and pairing difference V-GBFs $\vec{\mu}(\bm{x})=2 x_{\pi(1)}\cdot\vec{1}$ ($\omega\neq m$) or $\vec{\mu}(\bm{x})=2 x_{\pi(m)}\cdot\vec{1}$ ($\omega\neq 0$).
\end{fact}

\begin{definition}
[NSGIP \cite{Liu2013New}]\label{def1}
	A complex number is called a Gaussian integer if its real part and imaginary part are both integers.
	Define
	\begin{equation}
	{Q}(b_1,b_2,\dots,b_{q-1})=2^{q-1}+\sum_{p=1}^{q-1}2^{q-1-p}\xi^{b_p},\quad b_p\in\mathbb{Z}_4.
	\end{equation}
	as an one-to-one mapping from $ \mathbb{Z}_4^{q-1} $ to $ \mathcal{Q}_{q} $, which is a set consisting of $ 4^{q-1} $ Gaussian integers.

	For ${Q}_0= {Q}(b_1,b_2,\dots,b_{q-1})\in \mathcal{Q}_{q}  $ and ${Q}_1= {Q}(b'_1,b'_2,\dots,b'_{q-1})\in \mathcal{Q}_{q}  $, a pair of distinct
	Gaussian integers with identical magnitude, and which are not conjugate with each other, namely:
	\begin{equation}\label{NSGIP}
	|{Q}_0|=|{Q}_1|,\ {Q}_0\ne{Q}_1,\ \text{and}\ {Q}_0\ne \overline{Q}_{1},
	\end{equation}
$ ({Q}_0,{Q}_1) $ is called a non-symmetrical Gaussian integer pair (NSGIP).
\end{definition}

The  offset V-GBFs  $\vec{s}(\bm{x})$ of the generalized cases IV and V can be re-expressed  by the following $\vec{b}$-vectors derived from NSGIP.
\begin{fact}[Generalized cases IV-V \cite{Liu2013New}]\label{fact3}
Let $Q(b_1,b_2,\dots,b_{q-1})$ and $Q(b'_1,b'_2,\dots,b'_{q-1})$ be NSGIP.
Define the $\vec{b}$-vectors by $\vec{b}=(0,b_{1},b_{2},\dots,b_{q-1})^{T}$ and $\vec{b}'=(0,b'_{1},b'_{2},\dots,b'_{q-1})^{T}$.
$\{{F}(y) , {G}(y)\}$  forms a $4^{q}$-QAM GCP of length $2^m$ if the offset V-GBFs 	
\[ \vec{s}(\bm{x})=\vec{b}+( \vec{b}'-\vec{b})\cdot x_{\pi(\upsilon)}\]
for $2\le\upsilon\le m-1$, or
\[ \vec{s}(\bm{x})=\vec{b}+( \vec{b}'-\vec{b})\cdot x_{\pi(\upsilon_{1})}+( -\vec{b}'-\vec{b})\cdot x_{\pi(\upsilon_{2})}\]
for $1\le\upsilon_{1}\le m-2 $, $\upsilon_{1}+2\le\upsilon_{2}\le m $, and pairing difference V-GBFs $\vec{\mu}(\bm{x})=2 x_{\pi(1)}\cdot\vec{1} $ or $ 2  x_{\pi(m)}\cdot\vec{1}$.
\end{fact}

\subsection{Coefficient Matrix of Offset}

All the offsets $\vec{s}(\bm{x})$ in the generalized cases I-V and the new constructions in this paper can be represented by the form
\[
\vec{s}(\bm{x})=
\begin{pmatrix}
s^{(0)}(\bm{x})\\
s^{(1)}(\bm{x})\\
\vdots\\
s^{(q-1)}(\bm{x})\\
\end{pmatrix}
=
\begin{pmatrix}
\begin{array}{cccc}
c_{0,0}& c_{0,1} & \dots & c_{0,m} \\
c_{1,0}& c_{1,1} & \dots & c_{1,m} \\
\vdots& \vdots & \ddots & \vdots \\
c_{q-1,0}& c_{q-1,1} & \dots & c_{q-1,m} \\
\end{array}
\end{pmatrix}
\begin{pmatrix}
\begin{array}{c}
1\\
x_{\pi(1)}\\
\vdots\\
x_{\pi(m)}\\
\end{array}
\end{pmatrix},
\]
or alternatively,
$$\vec{s}(\bm{x})=(\vec{c},\vec{c}_{1},\dots,\vec{c}_{m})(1, x_{\pi(1)},\dots,x_{\pi(m)})^{T}.$$
We call $\mathcal{S}=\{c_{p,j}\}_{q\times m}=\left(
	\vec{c},\vec{c}_{1},\dots,\vec{c}_{m}
	\right)$  the {\em coefficient matrix} of $\vec{s}(\bm{x})$.

In order to get a unified representation, we define {\em redundant} vectors $\vec{c}_{0}$ and $\vec{c}_{m+1}$ corresponding to fake variables $x_{\pi(0)}$ and $x_{\pi(m+1)}$. Then the offsets $\vec{s}(\bm{x})$ can be represented by the {\em extended coefficient matrix}:
$$\vec{s}(\bm{x})=(\vec{c},\vec{c}_{0},\vec{c}_{1},\dots,\vec{c}_{m},\vec{c}_{m+1}
	)(1,x_{\pi(0)},x_{\pi(1)},\dots,x_{\pi(m)},x_{\pi(m+1)})^{T}.$$
Then the non-zero columns in the extended coefficient matrices for the generalized cases I-V are given below.
\begin{itemize}
	\item
	I-III:
	${\vec{c}}=\vec{d}_{0}$,
	$\vec{c}_{\omega}=\vec{d}_{1}$,
	$\vec{c}_{\omega+1}=\vec{d}_{2}$, for $0\le \omega\le m$;
\item IV:
	$\vec{c}=\vec{b}$ and $\vec{c}_{\upsilon}=\vec{b}'-\vec{b}$, for $2\leq{\upsilon}\leq m-1$;
	\item V: $\vec{c}=\vec{b}$, $\vec{c}_{\upsilon_{1}}=\vec{b}'-\vec{b}$ and $\vec{c}_{\upsilon_{2}}=-\vec{b}'-\vec{b}$, for  $1\le{\upsilon_{1}}\le m-2 $, $\upsilon_{1}+2\le{\upsilon_{2}}\le m $.
\end{itemize}

Thus, there are at most 2 non-zero columns $\vec{c}_j$ ($1\leq j\leq m$) in coefficient matrices of $\vec{s}(\bm{x})$ for the generalized cases I-V. In the next section, we will propose new GCPs and GCSs, by constructing more flexible  coefficient matrices of $\vec{s}(\bm{x})$.

\section{Main Results}
In this section, we will present two new constructions of GCPs over $4^{q}$-QAM. The  generalized cases I-III \cite{Li2010A} and the generalized cases IV-V \cite{Liu2013New} are
special cases of our first and second constructions, respectively. Those $\vec{d}$-vectors and $\vec{b}$-vectors derived from  the set $\mathcal{C}$ and NSGIP are still the
ingredients to construct the offset V-GBFs $\vec{s}(\bm{x})$ and pairing difference V-GBFs $\vec{\mu}(\bm{x})$. However, different from the generalized cases I-V, the new proposed $\vec{s}(\bm{x})$ and $\vec{\mu}(\bm{x})$ relate to the factorization of the integer $q$.

We first introduce the concept of mixed  radix numeral systems, in which the numerical base varies from position to position.
\begin{definition}[Mixed radix representation]\label{[Mixed radix representation}
	Suppose that $q=\prod_{k=1}^{t}q_{k}$, where $q_{k}$ ($1\le k\le t$) are positive integers larger than 1.
Then any integer $p$
($0\le p\le q-1$) can be uniquely represented by mixed radix representation as
$$p=\rho_{1}(p)+\rho_{2}(p)q_{1}+\rho_{3}(p)q_{1}q_{2}+\cdots+\rho_{t}(p)\prod_{k=1}^{t-1}q_{k}$$
for $0\le \rho_{k}(p)\le q_{k}-1$  $ (1\le k\le t)$, denoted by
$$p=(\rho_{t}(p),\rho_{t-1}(p),\dots, \rho_{1}(p))_{q_{t}q_{t-1}\dots q_{1}}.$$
\end{definition}

\begin{example}
The most familiar example  of mixed radix representation is in timekeeping system. We have 7 days in a week, 24 hours in a day, 60 minutes in a hour, and 60 seconds in a minute. The system for describing the  $604800=60\times 60 \times 24 \times 7$ runs as the follows table.
\[
\begin{tabular}{|c|c|c|c|c|}
	\hline
	        {Radix}         &  $q_1=60$  &  $q_2=60$  &  $q_3=24$  &   $q_4=7$   \\ \hline
	    {Denomination}      & second & minute &  hour  &   day   \\ \hline
	{Place value (seconds)} &  $1$   &  $q_1=60$  & $q_1q_2=3600$ & $q_1q_1q_3=86400$ \\ \hline
	{Number $p=323516$}&${\rho}_{1}(p)=56$&${\rho}_{2}(p)=51$&${\rho}_{3}(p)=17$&${\rho}_{4}(p)=3$\\\hline
\end{tabular}
\]

The $323516$th second can be represented by
$$323516=56+51\times 60+17\times 60\times 60+3\times 60\times 60\times 24.$$
Then $323516=(3,17,51,56)_{7,24,60,60}$ would be interpreted as $17:51:56$ on Wednesday.
\end{example}

\subsection{The First Construction}

Let $q=q_1\times q_2\times\cdots\times q_t$ be an ordered factorization of $q$, where $q_{k}\ge2 \;(1\le k\le t)$ are positive integers. Then any integer $p$ $(0\le p\le q-1)$ can be represented by  $(\rho_{t}(p),\rho_{t-1}(p),\dots, \rho_{1}(p))_{q_{t}q_{t-1}\dots q_{1}}$ according to the mixed radix representation. A tabular summary is given below.
\[\begin{tabular}{|c|c|c|c|c|c|}
\hline
{Radix}&$q_{1}$&$q_{2}$&$q_{3}$&$\dots$&$q_{t}$\\
\hline
{Place value}& $1$ &$q_{1}$&$q_{1}\cdot q_{2}$&$\dots$&$\prod_{k=1}^{t-1}q_{k}$\\\hline
{Number} $p$ &${\rho}_{1}(p)$&${\rho}_{2}(p)$&${\rho}_{3}(p)$&$\dots$&${\rho}_{t}(p)$\\
\hline
\end{tabular}\]

Recall the set $\mathcal{C}$ given in Definition \ref{Set-C}. We define the new $\vec{d}$-vectors $\vec{d_{i}}^{(k)}$ with respected to the ordered factorization of $q$.
\begin{definition}[$\vec{d}$-vector $\vec{d_{i}}^{(k)}$]\label{def_d}
For $1\leq k\leq t$, $1\le p_{k}\le q_{k}-1$, We arbitrarily choose ${\underline{d}^{(k,p_{k})}}=\left({d_{0}^{(k,p_{k})}}, {d_{1}^{(k,p_{k})}}, {d_{2}^{(k,p_{k})}}\right)\in \mathcal{C}$ and always define  ${\underline{d}^{(k,0)}}=(0,0,0)$. Then, for $1\le k\le t$ and $i=0,1,2$,  we define the $\vec{d}$-vectors  by
$$\vec{d_{i}}^{(k)}=\left(d_{i}^{(k,{\rho}_{k}(0))},d_{i}^{(k,{\rho}_{k}(1))},\dots,d_{i}^{(k,{\rho}_{k}(q-1))}\right)^{T}.$$
\end{definition}

According to the definition of the set $\mathcal{C}$, we have
 $2\vec{d_{0}}^{(k)}+\vec{d_{1}}^{(k)}+\vec{d_{2}}^{(k)}=\vec{0}$. Moreover, If the factorization of $q$ is trivial, we have $t=1$, and the $\vec{d}$-vectors given  in Definition
 \ref{def_d} agree with the $\vec{d}$-vectors for the generalized cases I-III shown in Fact \ref{fact2}.
\begin{example}\label{exam1}
For $q=6$ and  ordered factorization $q=q_1\times q_2=3\times2$, the mixed radix representation of $p=(\rho_2(p), \rho_1(p))_{2,3}$ is given in the left side of the following table. Moreover, if we choose
\[
		\underline{d}^{(1,1)}=(0,1,3),\quad
		\underline{d}^{(1,2)}=(1,0,2),\quad
		\underline{d}^{(2,1)}=(3,1,2),\quad
		\underline{d}^{(1,0)}=\underline{d}^{(2,0)}=(0,0,0),
		\]		
which can be read row by row in the middle and right sides of the table. Then we can read the corresponding $\vec{d}$-vectors $\vec{{d}_i}^{(1)}$ and $\vec{{d}_i}^{(2)}$ column by column. For example, $\vec{{d}_2}^{(1)}=(0,3,2,0,3,2)^{T}$, $\vec{{d}_0}^{(2)}=(0,0,0,3,3,3)^{T}$.
\[\begin{tabular}{|c|cc||c|ccc||c|ccc|}
\hline
$ p $&${\rho}_{1}(p)$&${\rho}_{2}(p)$&$\underline{d}^{(1,{\rho}_1(p))}$&$\vec{{d}_0}^{(1)}$&$\vec{{d}_1}^{(1)}$&$\vec{{d}_2}^{(1)}$&$\underline{d}^{(2,{\rho}_2(p))}$&$\vec{{d}_0}^{(2)}$&$\vec{{d}_1}^{(2)}$&$\vec{{d}_2}^{(2)}$\\\hline
$ 0 $&$ 0 $&$ 0 $&$\underline{d}^{(1,0)}$&$ 0 $&$ 0 $&$ 0 $&$\underline{d}^{(2,0)}$&$ 0 $&$ 0 $&$ 0 $\\
$ 1 $&$ 1 $&$ 0 $&$\underline{d}^{(1,1)}$&$ 0 $&$ 1 $&$ 3 $&$\underline{d}^{(2,0)}$&$0$&$0$&$0$\\
$ 2 $&$ 2 $&$ 0 $&$\underline{d}^{(1,2)}$&$1$&$0$&$2$&$\underline{d}^{(2,0)}$&$0$&$0$&$0$\\
$ 3 $&$ 0 $&$ 1 $&$\underline{d}^{(1,0)}$&$0$&$0$&$0$ &$\underline{d}^{(2,1)}$&$3$&$1$&$2$\\
$ 4 $&$ 1 $&$ 1 $&$\underline{d}^{(1,1)}$&$0$&$1$&$3$&$\underline{d}^{(2,1)}$&$3$&$1$&$2$\\
$ 5 $&$ 2 $&$ 1 $&$\underline{d}^{(1,2)}$&$1$&$0$&$2$&$\underline{d}^{(2,1)}$&$3$&$1$&$2$\\
\hline
\end{tabular}\]		
\end{example}

\begin{theorem}\label{thm_main_1}
	Suppose that factorization of $q$  and $\vec{d}$-vectors $\vec{d_{i}}^{(k)}$ are given above.
For arbitrary ordered position set $\{\omega_{1},\omega_{2},\dots,\omega_{t}\}\subset \{0,1,\dots,m\}$,
sequences over $4^{q}$-QAM of length $2^m$ associated with V-GBFs in form (\ref{equation_QAMGolaypair}) form a  GCP  if the offset V-GBFs $\vec{s}(\bm{x})$ satisfy
$$\vec{s}(\bm{x})=\sum_{k=1}^{t}\left(\vec{d_{1}}^{(k)}x_{\pi(\omega_{k})}+\vec{d_{2}}^{(k)}x_{\pi({\omega_{k}}+1)}+\vec{d_{0}}^{(k)}\right)$$
and the pairing difference V-GBFs $\vec{\mu}(\bm{x})$  satisfy
$$ \vec{\mu}(\bm{x})=\left\{
	\begin{aligned}
	&2x_{\pi(1)}\cdot\vec{1}+\vec{d_{1}}^{(k)},&{\exists k,\,\omega_{k}=0};\\
	&2x_{\pi(1)}\cdot\vec{1},&\text{otherwise};
	\end{aligned}\right.
	\quad\text{or}\;\left\{
	\begin{aligned}
	&2x_{\pi(m)}\cdot\vec{1}+d_{2}^{(k)},&\exists k,\,\omega_{k}=m;\\
	&2x_{\pi(m)}\cdot\vec{1},&\text{otherwise}.
	\end{aligned}\right. $$
\end{theorem}

The validity of Theorem \ref{thm_main_1} will be proved in Subsection 7.2.

Since the $\vec{d}$-vectors in Theorem \ref{thm_main_1} coincide with $\vec{d}$-vectors in the generalized cases I-III if the factorization of $q$ is trivial,
the generalized cases I-III in \cite{Li2010A} are special cases in our Theorem \ref{thm_main_1} for $t=1$. For composite number $q$ and non-trivial factorization,
new GCPs and GCSs over QAM are constructed.

\begin{example}\label{exam2}
For $q=6=q_1\times q_2=3\times 2$, $\omega_{1}=m$ and $1\le \omega_{2}=\omega \leq m-2$, the offset and the pairing difference V-GBFs in Theorem \ref{thm_main_1} are given below.
\begin{eqnarray*}
\vec{s}(\bm{x})
&=&\vec{d_{1}}^{(2)}\cdot x_{\pi(\omega)}+\vec{d_{2}}^{(2)}\cdot x_{\pi(\omega{+}1)}+\vec{d_{0}}^{(2)}
+\vec{d_{1}}^{(1)}\cdot x_{\pi(m)}+\vec{d_{0}}^{(1)},\\
		\vec{\mu}(\bm{x})&=&2 x_{\pi(1)}\cdot\vec{1}
	\quad\text{or}\ \ 2x_{\pi(m)}\cdot\vec{1}+\vec{d_2}^{(1)}.
	\end{eqnarray*}
The details of the offset V-GBFs are shown by the following table.
\[\begin{tabular}{|c|cc|ll|c|}
\hline
$ p $&${\rho}_{1}(p)$&${\rho}_{2}(p)$&${d}_i^{(1,{\rho}_1(p))}$&${d}_i^{(2,{\rho}_2(p))}$&{offset} : $s^{(p)}(\bm{x})$    \\\hline
$ 0 $&$ 0 $&$ 0 $&${d}_i^{(1,0)}=0$&${d}_i^{(2,0)}=0$&$0$\\
$ 1 $&$ 1 $&$ 0 $&${d}_i^{(1,1)}$&${d}_i^{(2,0)}=0$&$d_{1}^{(1,1)}x_{\pi(m)}+d_{0}^{(1,1)}$\\
$ 2 $&$ 2 $&$ 0 $&${d}_i^{(1,2)}$&${d}_i^{(2,0)}=0$&$d_{1}^{(1,2)}x_{\pi(m)}+d_{0}^{(1,2)}$\\
$ 3 $&$ 0 $&$ 1 $&${d}_i^{(1,0)}=0$&${d}_i^{(2,1)}$&$d_{1}^{(2,1)}x_{\pi(\omega)}+d_{2}^{(2,1)}x_{\pi(\omega+1)}+d_{0}^{(2,1)}$\\
$ 4 $&$ 1 $&$ 1 $&${d}_i^{(1,1)}$&${d}_i^{(2,1)}$&
$d_{1}^{(2,1)}x_{\pi(\omega)}+d_{2}^{(2,1)}x_{\pi(\omega+1)}+d_{0}^{(2,1)}+d_{1}^{(1,1)}x_{\pi(m)}+d_{0}^{(1,1)}$\\
$ 5 $&$ 2 $&$ 1 $&${d}_i^{(1,2)}$&${d}_i^{(2,1)}$&
$d_{1}^{(2,3)}x_{\pi(\omega)}+d_{2}^{(2,1)}x_{\pi(\omega+1)}+d_{0}^{(2,1)}+d_{1}^{(1,2)}x_{\pi(m)}+d_{0}^{(1,2)}$\\
\hline
\end{tabular}\]	
Moreover, if we choose $\underline{d}^{(k)}$ ($k=1,2$) from Example \ref{exam1}, the non-zero columns of corresponding coefficient matrix of offset is given by
\[
\mathcal{S}
=\bordermatrix{
&\vec{c}&\dots& \vec{c}_{\omega} & \vec{c}_{\omega+1} & \dots & \vec{c}_{m}\cr
&0&\dots&0&0& \dots &0\cr
&0&\dots&0&0& \dots &1\cr
&1&\dots&0&0& \dots &0\cr
&3&\dots&1&2& \dots &0\cr
&3&\dots&1&2& \dots &1\cr
&0&\dots&1&2& \dots &0\cr},
\]
from which we can see there are 3 non-zero $\vec{c}_j$ for $1\leq j\leq m$. This demonstrates that the construction in Theorem \ref{thm_main_1} produces new GCSs over QAM.
\end{example}

\subsection{The Second Construction}

In this subsection, we slightly modify the conditions in the first construction, and obtain our second construction which includes the generalized cases IV-V in \cite{Liu2013New} as special cases.

Let $q=q_0\times q_1\times\cdots\times q_t$ be an ordered factorization of $q$, where  ${q_{0}}\ge3$ and $q_{k}\ge2$  $(1\le k\le t)$ are positive integers.
Then any integer $p$ $(0\le p\le q-1)$ can be represented by  $(\rho'_{t}(p),\rho'_{t-1}(p),\dots, \rho'_{0}(p))_{q_{t}q_{t-1}\dots q_{0}}$ according to the mixed radix representation. A tabular summary is given below.
\[\begin{tabular}{|c|c|c|c|c|c|}
		\hline
		{Radix}&$q_{0}$&$q_{1}$&$q_{2}$&$\dots$&$q_{t}$\\
		\hline
		{Place value}& $1$ &$q_{0}$&$q_{0}\cdot q_{1}$&$\dots$&$\prod_{k=0}^{t-1}q_{k}$\\\hline
		{Number}&${\rho}'_{0}(p)$&${\rho}'_{1}(p)$&${\rho}'_{2}(p)$&$\dots$&${\rho}'_{t}(p)$\\
		\hline
\end{tabular}\]

Recall the set $\mathcal{C}$ in Definition \ref{Set-C} and NSGIP in Definition \ref{def1}. We define the $\vec{d}$-vectors and $\vec{b}$-vectors with respected to the ordered factorization of $q$ in this subsection.

\begin{definition}[$\vec{d}$-vector $\vec{d_{i}}^{[k]}$ and $\vec{b}$-vector]\label{p'_decomposition}
For $1\leq k\leq t$, $1\le p_{k}\le q_{k}-1$, we arbitrarily choose ${\underline{d}^{(k,p_{k})}}=\left({d_{0}^{(k,p_{k})}}, {d_{1}^{(k,p_{k})}}, {d_{2}^{(k,p_{k})}}\right)\in \mathcal{C}$, and always define ${\underline{d}^{(k,0)}}=(0,0,0)$. Then, for $1\le k\le t$ and $i=0,1,2$,  we define the $\vec{d}$-vectors  by
$$\vec{d_{i}}^{[k]}=\left(d_{i}^{(k,{\rho}'_{k}(0))},d_{i}^{(k,{\rho}'_{k}(1))},\dots,d_{i}^{(k,{\rho}'_{k}(q-1))}\right)^{T}.$$
	
	Suppose that $Q_0= Q(b_1,b_2,\dots,b_{q_{0}-1})$ and $Q_1= Q(b'_1,b'_2,\dots,b'_{q_{0}-1})$ are NSGIP over $\mathcal{Q}_{q_0}$ introduced in Definition \ref{def1}. Let $b_{0}=b'_{0}=0$. We define the $\vec{b}$-vectors  by $$\vec{b}=\left(b_{{\rho}'_{0}(0)},b_{{\rho}'_{0}(1)},\dots,b_{{\rho}'_{0}(q-1)}\right)^{T} \mbox{and}\
 \vec{b}'=\left(b'_{{\rho}'_{0}(0)},b'_{{\rho}'_{0}(1)},\dots,b'_{{\rho}'_{0}(q-1)}\right)^{T}.$$
\end{definition}

\begin{example}\label{exam3}
For $q=6=q_0\times q_1=3\times2$, the mixed radix representation of $p=(\rho'_1(p), \rho'_0(p))_{2,3}$ is shown in the left side of the following table. If we choose
$\underline{d}^{(1,1)}=(3,1,2)$, which can be read row by row in the middle of the table,  then we can read the corresponding $\vec{d}$-vectors column by column.

One typical NSGIP over $\mathcal{Q}_{3}$ is ${Q}_0= {Q}(0,2)=5$ and ${Q}_1= {Q}(1,1)=4+3\xi$. The corresponding $\vec{b}$-vectors are given in the right side of the table.
\[\begin{tabular}{|c|cc||c|ccc||c|c|c|c|}
	\hline
	$ p $&${\rho}'_{0}(p)$&${\rho}'_{1}(p)$&$\underline{d}^{(1,{\rho}'_1(p))}$&$\vec{{d}_0}^{[1]}$&$\vec{{d}_1}^{[1]}$&$\vec{{d}_2}^{[1]}$&${b}_{{\rho}'_{0}(p)}$&$\vec{b}$&${b}'_{{\rho}'_{0}(p)}$&$\vec{b}'$\\\hline
	$ 0 $&$ 0 $&$ 0 $&$\underline{d}^{(1,0)}$&$ 0 $&$ 0 $&$ 0 $&$b_{0}$&$ 0 $&$b'_{0}$&$ 0 $\\
	$ 1 $&$ 1 $&$ 0 $&$\underline{d}^{(1,0)}$&$ 0 $&$ 0 $&$ 0 $&$b_{1}$&$0$&$b'_{1}$&$1$\\
	$ 2 $&$ 2 $&$ 0 $&$\underline{d}^{(1,0)}$&$0$&$0$&$0$&$b_{2}$&$2$&$b'_{2}$&$1$\\
	$ 3 $&$ 0 $&$ 1 $&$\underline{d}^{(1,1)}$&$3$&$1$&$2$ &$b_{0}$&$0$&$b'_{0}$&$0$\\
	$ 4 $&$ 1 $&$ 1 $&$\underline{d}^{(1,1)}$&$3$&$1$&$2$&$b_{1}$&$0$&$b'_{1}$&$1$\\
	$ 5 $&$ 2 $&$ 1 $&$\underline{d}^{(1,1)}$&$3$&$1$&$2$&$b_{2}$&$2$&$b'_{2}$&$1$\\
	\hline
	\end{tabular}\]		
\end{example}

\begin{theorem}\label{thm_main_2}
Suppose that factorization of $q$, $\vec{d}$-vectors  $\vec{d_{i}}^{[k]}$ and $\vec{b}$-vectors $\vec{b}$, $\vec{b}'$
are given above. For arbitrary ordered position set $\{\omega_{1},\omega_{2},\dots,\omega_{t}\}\subset \{0,1,\dots,m\}$,
sequences over $4^{q}$-QAM of length $2^m$ associated with V-GBFs in form (\ref{equation_QAMGolaypair}) form a  GCP  if the offset V-GBFs $\vec{s}(\bm{x})$ satisfy
\begin{itemize}
\item{Case (a)}: for $2\leq {\upsilon}\leq m-1$,	
$$\vec{s}(\bm{x})=\sum_{k=1}^{t}\left(\vec{d_{1}}^{[k]}\cdot x_{\pi(\omega_{k})}+\vec{d_{2}}^{[k]}\cdot x_{\pi({\omega_{k}}+1)}+\vec{d_{0}}^{[k]}\right)+\left(\left(\vec{b}'-\vec{b}\right)\cdot x_{\pi(\upsilon)}+\vec{b}\right),$$
\item{Case (b)}: for $1\le{\upsilon_{1}}\le m-2 $, $\upsilon_{1}+2\le{\upsilon_{2}}\le m $,	
$$\vec{s}(\bm{x})=\sum_{k=1}^{t}\left(\vec{d_{1}}^{[k]}\cdot x_{\pi(\omega_{k})}+\vec{d_{2}}^{[k]}\cdot x_{\pi({\omega_{k}}+1)}+\vec{d_{0}}^{[k]}\right)+\left(\left(\vec{b}'-\vec{b}\right)\cdot x_{\pi(\upsilon_{1})}+\left(-\vec{b}'-\vec{b}\right)\cdot x_{\pi(\upsilon_{2})}+\vec{b}\right),$$
\end{itemize}
and the pairing difference V-GBFs $\vec{\mu}(\bm{x})$  satisfy
$$ \vec{\mu}(\bm{x})=\left\{
			\begin{aligned}
			&2x_{\pi(1)}\cdot\vec{1}+\vec{d_{1}}^{[k]},&{\exists k,\,\omega_{k}=0};\\
			&2x_{\pi(1)}\cdot\vec{1},&\text{otherwise};
			\end{aligned}\right.
			\quad\text{or}\;\left\{
			\begin{aligned}
			&2x_{\pi(m)}\cdot\vec{1}+d_{2}^{[k]},&\exists k,\,\omega_{k}=m;\\
			&2x_{\pi(m)}\cdot\vec{1},&\text{otherwise}.
			\end{aligned}\right. $$
\end{theorem}

We will prove Theorem \ref{thm_main_2} in Subsection 7.2.

Since the $\vec{b}$-vectors in Theorem \ref{thm_main_2} coincide with $\vec{b}$-vectors in the generalized cases IV-V if the factorization of $q=q_0$ is trivial,
the generalized cases IV-V in \cite{Liu2013New} are special cases of our Theorem \ref{thm_main_2} for $t=0$. For composite number $q$ and non-trivial factorization,
new GCPs and GCSs over QAM are constructed. We give an example of Case (b) to illustrate it.


\begin{example}
For $q=6=q_0\times q_1=3\times2$, $1\le\omega_{1}=\omega \le m-1$,
 the offset and the pairing difference V-GBFs of Case (b) in Theorem \ref{thm_main_2} are given below.
\begin{eqnarray*}
\vec{s}(\bm{x})
			&=&\vec{d_{1}}^{(1)}\cdot x_{\pi(\omega)}+\vec{d_{2}}^{(1)}\cdot x_{\pi(\omega{+}1)}+\vec{d_{0}}^{(1)}
			+(\vec{b}'-\vec{b})x_{\pi(\upsilon_{1})}+(-\vec{b}'-\vec{b})x_{\pi(\upsilon_{2})}+\vec{b},\\
		\vec{\mu}(\bm{x})&=&2x_{\pi(1)}\cdot\vec{1}
			\quad\text{or}\quad
			2x_{\pi(m)}\cdot\vec{1}.
	\end{eqnarray*}
The details of the offset V-GBFs can be shown by the following table.
	\[\begin{tabular}{|c|cc|l|c|}
	\hline
	$ p $&${{\rho}'_{0}(p)}$&${\rho}'_{1}(p)$&${d}_i^{(1,{\rho}'_1(p))}$&{offset} : ${s}^{(p)}(\bm{x})  $    \\\hline
	$ 0 $&$ 0 $&$ 0 $&${d}_i^{(1,0)}=0$&$0$\\
	$ 1 $&$ 1 $&$ 0 $&${d}_i^{(1,0)}=0$&$(b'_{1}-b_{1})x_{\pi(\upsilon_{1})}+(-b'_{1}-b_{1})x_{\pi(\upsilon_{2})}+{b_1}$\\
	$ 2 $&$ 2 $&$ 0 $&${d}_i^{(1,0)}=0$&$(b'_{2}-b_{2})x_{\pi(\upsilon_{1})}+(-b'_{2}-b_{2})x_{\pi(\upsilon_{2})}+{b_2}$\\
	$ 3 $&$ 0 $&$ 1 $&${d}_i^{(1,1)}$&$d_{1}^{(1,1)}x_{\pi(\omega)}+d_{2}^{(1,1)}x_{\pi(\omega+1)}+d_{0}^{(1,1)}$\\
	$ 4 $&$ 1 $&$ 1 $&${d}_i^{(1,1)}$&$d_{1}^{(1,1)}x_{\pi(\omega)}+d_{2}^{(1,1)}x_{\pi(\omega+1)}+(b'_{1}-b_{1})x_{\pi(\upsilon_{1})}+(-b'_{1}-b_{1})x_{\pi(\upsilon_{2})}+d_{0}^{(1,1)}+{b_1}$\\
	$ 5 $&$ 2 $&$ 1 $&${d}_i^{(1,1)}$&$d_{1}^{(1,1)}x_{\pi(\omega)}+d_{2}^{(1,1)}x_{\pi(\omega+1)}+(b'_{2}-b_{2})x_{\pi(\upsilon_{1})}+(-b'_{2}-b_{2})x_{\pi(\upsilon_{2})}+d_{0}^{(1,1)}+{b_2}$\\
	\hline	\end{tabular}\]	

If we choose $\underline{d}^{[1]}$, $\vec{b}$ and $\vec{b}'$ from Example \ref{exam3} and $\upsilon_1=\omega$,  the non-zero columns of coefficient matrix of offset is given by
\[
\mathcal{S}
=\bordermatrix{
	&\vec{c}&\dots& \vec{c}_{\omega} & \vec{c}_{\omega+1} & \dots & \vec{c}_{\upsilon_{2}} & \dots \cr
	&0&\dots&0&0& \dots  &0& \dots \cr
	&0&\dots&3&3& \dots  &3& \dots \cr
	&2&\dots&1&1& \dots  &1& \dots \cr
	&3&\dots&1&2& \dots  &0& \dots \cr
	&3&\dots&0&1& \dots  &3& \dots \cr
	&1&\dots&2&3& \dots  &1& \dots \cr}
\]
from which we can see there are 3 non-zero $\vec{c}_j$ for $1\leq j\leq m$. Moreover, if $\upsilon_1, \upsilon_2\notin \{\omega, \omega+1\}$, there will be 4 non-zero $\vec{c}_j$  for $1\leq j\leq m$ in the coefficient matrix of offset.
This demonstrate that the construction in Theorem \ref{thm_main_2} produce new GCSs over QAM.
\end{example}

\section{Enumerations}

The number of the  GCSs over $4^{q}$-QAM of length $2^{m}$ constructed in Theorems \ref{thm_main_1} and \ref{thm_main_2} is equal to the product of the number of the
standard GCSs ${f}(\bm{x})$ over QPSK and the number of the compatible offsets $\vec{s}(\bm{x})$, i.e.,
$$\#\{\vec{s}(\bm{x})\}\times \#\{{f}(\bm{x})\}.$$
It is well known that the number of the standard GCSs over QPSK is given by $\#\{{f}(\bm{x})\}=(m!/2)4^{(m+1)}$. So the enumeration of the GCSs is determined by the number of the compatible offsets. Moreover, each offset  can be uniquely represented by its coefficient matrix, so we have $\#\{\vec{s}(\bm{x})\}= \#\{\mathcal{S}\}$.

For given $q$, it was shown that the number of the compatible offsets in the generalized cases I-III and cases IV-V are a linear polynomial of $m$ \cite{Li2010A} and  a quadratic polynomial of $m$ \cite{Liu2013New}, respectively. Notice that the coefficient matrices of  offsets  in the generalized cases I-V have at most 2 non-zero columns $\vec{c}_j$ ($1\leq j\leq m$). For $q=4$ and $q=6$, by studying the coefficient matrices  with three and four non-zero columns $\vec{c}_j$, a lower bound of the numbers of new offsets other than the generalized cases I-V  is shown in Table \ref{table-3}. Moreover, for $q=q_{1}\times q_{2}\times \dots\times q_{t}$ ($q_k>1$), we show that the number of the new offsets in Theorem 1 is lower bounded by a polynomial of $m$ with degree $t$.

 \begin{table}
	\centering
	\caption{Comparisons of the numbers of the compatible offsets}	
\begin{threeparttable}
 \begin{tabular}{|c|c|c|c|}\hline
           & $ q=4 $& $q=6$ &$q=q_{1}\times q_{2}\times \dots\times q_{t}$\\\hline
\thead{The generalized\\ cases I-III \cite{Li2010A}}& $ 4032m+4040 $& $ 1047552m+1047584 $& $f_1(m)^*\  (deg(f_1)=1)$\\ \hline
\thead{The generalized\\cases IV-V \cite{Liu2013New}}   &$ 14(m^2-m-2)$&$ 584(m^2-m-2) $  &$f_2(m)^\dag\  (deg(f_2)=2)$  \\ \hline
\thead{New constructions\\in this paper}  &$\geq 100(m^2-m-2)$ &   $\geq (3700+20m)(m^2-m-2)$ &   $\geq f_t(m)^\ddag\  (deg(f_t)=t)$  \\ \hline
\end{tabular} \label{table-3}
\begin{tablenotes}\footnotesize
	\item[*]
	Linear polynomial $f_1(m)$ was given in \cite{Li2010A}.
	\item[\dag]
	Quadratic polynomial $f_2(m)$ was given in \cite{Liu2013New}.
    \item[\ddag]
	Polynomial $f_t(m)$  of degree $t$ is given in formula (\ref{lower-bound}).
\end{tablenotes}	
\end{threeparttable}		
\end{table}

Before listing the new compatible offsets, we classify the set $\mathcal{C}$ in Definition \ref{Set-C} into four classes according to the values of $d_{1}$ and $d_{2}$.
\begin{eqnarray*}
\mathcal{C}_{1}&=&\{(1,1,1),(3,1,1),(0,1,3),(2,1,3),(0,2,2),(2,2,2),(0,3,1),(2,3,1),(1,3,3),(3,3,3)\},\\
\mathcal{C}_{2}&=&\{(1,0,2),(3,0,2)\},\\
\mathcal{C}_{3}&=&\{(1,2,0),(3,2,0)\},\\
\mathcal{C}_{4}&=&\{(0,0,0),(2,0,0)\}.
\end{eqnarray*}
Then we have $d_{1}, d_{2}\neq 0$ for $\underline{d}=(d_{0}, d_{1}, d_{2})\in \mathcal{C}_{1}$,  $d_{1}=0, d_{2}\neq 0$ for $\underline{d}\in \mathcal{C}_{2}$, $d_{1}\neq 0, d_{2}=0$ for $\underline{d}\in \mathcal{C}_{3}$, and $d_{1}=d_{2}=0$ for $\underline{d}\in \mathcal{C}_{4}$.

\subsection{Enumerations  for $q=4$ and Generalization}

Since the number of the compatible offsets is independent of the permutation $\pi$, without loss of generality, we restrict $\pi$ to be the identity permutation in the rest of this section.

For the factorization $4=2\times2$, the mixed radix representation of $p=(\rho_2(p), \rho_1(p))_{2,2}$ is equivalent to the binary expansion of integer $p$.
The offset $\vec{s}(\bm{x})$ in Theorem \ref{thm_main_1} can be expressed by
\begin{equation}\label{enumeration-1} {s}^{(p)}(\bm{x})=\left(d_{1}^{(1,{\rho}_{1}(p))}x_{\omega_{1}}+d_{2}^{(1,{\rho}_{1}(p))}x_{{\omega_{1}}+1}+d_{0}^{(1,{\rho}_{1}(p))}\right)+\left(d_{1}^{(2,{\rho}_{2}(p))}x_{\omega_{2}}+d_{2}^{(2,{\rho}_{2}(p))}x_{{\omega_{2}}+1}+d_{0}^{(2,{\rho}_{2}(p))}\right)
\end{equation}

\begin{proposition}\label{prop-1}
For $q=4=2\times2$, $m\geq 3$, ${\underline{d}^{(1,1)}}, {\underline{d}^{(2,1)}}\in \mathcal{C}_{1}$,  the ordered pairs
$(\omega_{1}, \omega_{2})\neq (0, m), (m, 0)$ and $|\omega_{1}-\omega_{2}|\geq 2$, different choices  of $\left({\underline{d}^{(1,1)}}, {\underline{d}^{(2,1)}}, (\omega_{1}, \omega_{2})\right)$ in Theorem \ref{thm_main_1} determine different offsets with at least 3 non-zero columns $\vec{c}_j$ ($1\leq j\leq m$) in coefficient matrices.
\end{proposition}
\begin{proof}
The conditions $(\omega_{1}, \omega_{2})\neq (0, m), (m, 0)$ and $|\omega_{1}-\omega_{2}|\geq 2$ guarantee that there are at most one \lq fake\rq\ variable in $x_{\omega_{1}}, x_{\omega_{1}+1}, x_{\omega_{2}}, x_{\omega_{2}+1}$. Since none of $d_{1}^{(1,1)}, d_{2}^{(1,1)}, d_{1}^{(2,1)}, d_{2}^{(2,1)}$ equals 0,
the assertion follows immediately from the  Table \ref{table-1}, where show that there are at least 3 non-zero columns $\vec{c}_j$ ($1\leq j\leq m$) in coefficient matrices.
\end{proof}
\begin{table}
	\centering
	\caption{Non-zero columns in $\mathcal{S}$ for $q=4$}	
	\begin{tabular}{|c|lcccc|}
 \hline
$p$&$\vec{c}$ &$\vec{c}_{\omega_{1}}$&$\vec{c}_{\omega_{1}+1}$&$\vec{c}_{\omega_{2}}$&$\vec{c}_{\omega_{2}+1}$\\
 \hline
0&$0$&$0$          &$0$          &$0$          &$0$          \\
1&$d_{0}^{(1,1)}$ &$d_{1}^{(1,1)}$&$d_{2}^{(1,1)}$&$0$          &$0$          \\
2&$d_{0}^{(2,1)}$&$0$          &$0$          &$d_{1}^{(2,1)}$&$d_{2}^{(2,1)}$\\
3&$d_{0}^{(1,1)}+d_{0}^{(2,1)}$&$d_{1}^{(1,1)}$&$d_{2}^{(1,1)}$&$d_{1}^{(2,1)}$&$d_{2}^{(2,1)}$\\
 \hline
\end{tabular}\label{table-1}
\end{table}

\begin{proposition}\label{prop-2}
For $q=4$ and $m\geq 3$, Proposition \ref{prop-1} identifies $100(m+1)(m-2)$ compatible offsets other than the generalized cases I-V.
\end{proposition}
\begin{proof}
If $\omega_{1}=0$, we can select $\omega_{2}$ such that $2\leq \omega_{2}\leq m-1$.  If $\omega_{1}=m$, we can choose $\omega_{2}$ such that $1\leq \omega_{2}\leq m-2$. If $\omega_{1}\neq 0, m$, we can select $\omega_{2}$ such that $0\leq \omega_{2}\leq m$ and $\omega_{2}\neq \omega_{1}-1, \omega_{1}, \omega_{1}+1$. So there are a total of $(m+1)(m-2)$ ordered pairs $(\omega_{1}, \omega_{2})$.  For each ordered pair, there are $10\times10=100$ choices of ${\underline{d}^{(1,1)}}, {\underline{d}^{(2,1)}}$ such that ${\underline{d}^{(1,1)}}, {\underline{d}^{(2,1)}}\in \mathcal{C}_{1}$. Thus, the conditions in Proposition \ref{prop-1} identifies $100(m+1)(m-2)$  compatible offsets.
\end{proof}

This lower bound of the enumeration of new offsets for $q=4$ can be generalized to arbitrary $q$.
\begin{proposition}\label{prop-0}
For factorization $q=q_{1}\times q_{2}\times \dots\times q_{t}$ and $m\geq 2t-1$,  Theorem \ref{thm_main_1} identifies  at least
\begin{equation} \label{lower-bound}
(m+1)\cdot\frac{(m-t)!}{(m-2t+1)!}\prod_{k=1}^{t}(14^{q_{k}-1}-2\times2^{q_{k}-1})
\end{equation}
new compatible offsets with $2t-1$ or $2t$ non-zero columns $\vec{c}_j$ ($1\leq j\leq m$) in coefficient matrices.
\end{proposition}
\begin{proof}
See Appendix C.
\end{proof}

Proposition \ref{prop-0} shows that the number of the new  offsets in Theorem \ref{thm_main_1} is lower bounded by a polynomial of $m$ with degree $t$.

\subsection{Enumerations  for $q=6$}

We first elaborate four cases of offsets for $q=6$ from Theorems \ref{thm_main_1} and \ref{thm_main_2}.

Case (1): For factorization $q_1\times q_2=2\times 3$, 
the offset $\vec{s}(\bm{x})$ in Theorem \ref{thm_main_1} can be expressed in the form of (\ref{enumeration-1}).
Let the vectors ${\underline{d}^{(1,1)}}$, ${\underline{d}^{(2,1)}}$, ${\underline{d}^{(2,2)}}$ and the ordered pairs ($\omega_{1}, \omega_{2}$) satisfy the following conditions:
\begin{enumerate}
	\item[(1)] ${\underline{d}^{(1,1)}}\in \mathcal{C}_{1}, {\underline{d}^{(2,1)}}, {\underline{d}^{(2,2)}}\notin \mathcal{C}_{4}, ({\underline{d}^{(2,1)}}, {\underline{d}^{(2,2)}})\notin \left(\mathcal{C}_{2},\mathcal{C}_{2}\right)\bigcup\left(\mathcal{C}_{3},\mathcal{C}_{3}\right)$;
\item[(2)] $(\omega_{1}, \omega_{2})\neq (0, m), (m, 0)$ and $|\omega_{1}-\omega_{2}|\geq 2$.
\end{enumerate}

Case (2): For factorization $q_1\times q_2=3\times2$,
the offset $\vec{s}(\bm{x})$ in Theorem \ref{thm_main_1} can be expressed in the form of (\ref{enumeration-1}).
Let the vectors ${\underline{d}^{(1,1)}}$, ${\underline{d}^{(1,2)}}$, ${\underline{d}^{(2,1)}}$ and the ordered pairs ($\omega_{1}, \omega_{2}$) satisfy the following conditions:
\begin{enumerate}
	\item[(1)] ${\underline{d}^{(2,1)}}\in \mathcal{C}_{1}, {\underline{d}^{(1,1)}}, {\underline{d}^{(1,2)}}\notin \mathcal{C}_{4}, ({\underline{d}^{(1,1)}}, {\underline{d}^{(1,2)}})\notin \left(\mathcal{C}_{2},\mathcal{C}_{2}\right)\bigcup\left(\mathcal{C}_{3},\mathcal{C}_{3}\right)$;
\item[(2)] $(\omega_{1}, \omega_{2})\neq (0, m), (m, 0)$ and $|\omega_{1}-\omega_{2}|\geq 2$.
\end{enumerate}

Case (3): For factorization $q_0\times q_1=3\times2$,
the offset $\vec{s}(\bm{x})$ in Case (a) of Theorem \ref{thm_main_2}  can be expressed by
\begin{equation}\label{enumeration-2}
\quad{s}^{(p)}(\bm{x})=\left(d_{1}^{(1,{\rho}'_{1}(p))}x_{\omega}+d_{2}^{(1,{\rho}'_{1}(p))}x_{{\omega}+1}+d_{0}^{(1,{\rho}'_{1}(p))}\right)+\left((b'_{{\rho}'_{0}(p)}-b_{{\rho}'_{0}(p)})x_{\upsilon}+b_{{\rho}'_{0}(p)}\right).
\end{equation}
Let the vector ${\underline{d}^{(1,1)}}$ and the ordered pairs ($\omega, \upsilon$) satisfy the following conditions:
\begin{enumerate}
	\item[(1)] ${\underline{d}^{(1,1)}}\in \mathcal{C}_{1}$;
\item[(2)] $ 2\le \upsilon\le m-1 $, $1\le\omega\le m-1$, and  $\omega\neq \upsilon,\upsilon-1$.
\end{enumerate}

Case (4): For factorization $q_0\times q_1=3\times2$,
the offset $\vec{s}(\bm{x})$ in Case (b) of Theorem \ref{thm_main_2} can be expressed by
\begin{equation}\label{enumeration-3} \quad{s}^{(p)}(\bm{x})=\left(d_{1}^{(1,{\rho}'_{1}(p))}x_{\omega}+d_{2}^{(1,{\rho}'_{1}(p))}x_{{\omega}+1}+d_{0}^{(1,{\rho}'_{1}(p))}\right)+\left((b'_{{\rho}'_{0}(p)}{-}b_{{\rho}'_{0}(p)})x_{\upsilon_{1}}{+}({-}b'_{{\rho}'_{0}(p)}{-}b_{{\rho}'_{0}(p)})x_{\upsilon_{2}}{+}{b_{{\rho}'_{0}(p)}}\right).
	\end{equation}
Let the vector ${\underline{d}^{(1,1)}}$ and the ordered triples ($\omega, \upsilon_{1}, \upsilon_{2}$) satisfy the following conditions:
\begin{enumerate}
	\item[(1)] ${\underline{d}^{(1,1)}}\in \mathcal{C}_{1}$;
\item[(2)]  $1\le \upsilon_{1}\le m-2$,  $\upsilon_{1}+2\le \upsilon_{2}\le m$,  and $0\le\omega\le m$, $\omega\neq \upsilon_{1},\upsilon_{1}-1, \upsilon_{2}, \upsilon_{2}-1$.
\end{enumerate}

\begin{table}
	\centering
	\caption{Non-zero columns  in coefficient matrices of offsets for $q=6$}	
	\begin{tabular}{|c|cccc|cccc|cccc|}
		\hline
		&\multicolumn{4}{c|}{$\mathcal{S}_{1}$ in Case (1)}&\multicolumn{4}{c|}{$\mathcal{S}_{2}$ in Case (2)}&\multicolumn{4}{c|}{$\mathcal{S}_{3}$ and $\mathcal{S}_{4}$ in Case (3) and (4)}\\
		\cline{2-13}
		$p_{ }$&$\vec{c}_{\omega_{1}}$&$\vec{c}_{\omega_{1}+1}$&$\vec{c}_{\omega_{2}}$&$\vec{c}_{\omega_{2}+1}$&$\vec{c}_{\omega_{2}'}$&$\vec{c}_{\omega_{2}'+1}$&$\vec{c}_{\omega_{1}'}$&$\vec{c}_{\omega_{1}'+1}$&$\vec{c}_{\omega}$&$x_{\omega+1}$&$\vec{c}_{\upsilon_{1}}$&$\vec{c}_{\upsilon_{2}}$ \\
		\hline
		$0$&$0$          &$0$          &$0$          &$0$          &$0$          &$0$          &$0$          &$0$          &$0$          &$0$          &$0$          &$0$           \\
		$1$&$d_{1}^{(1,1)}$&$d_{2}^{(1,1)}$&$0$          &$0$          &$0$          &$0$          &$d_{1}^{(1,1)}$&$d_{2}^{(1,1)}$&$0$          &$0$          &$b'_{1}-b_{1}$&$-b'_{1}-b_{1}$\\
		$2$&$0$          &$0$          &$d_{1}^{(2,1)}$&$d_{2}^{(2,1)}$&$0$          &$0$          &$d_{1}^{(1,2)}$&$d_{2}^{(1,2)}$&$0$          &$0$          &$b'_{2}-b_{2}$&$-b'_{2}-b_{2}$\\
		$3$&$d_{1}^{(1,1)}$&$d_{2}^{(1,1)}$&$d_{1}^{(2,1)}$&$d_{2}^{(2,1)}$&$d_{1}^{(2,1)}$&$d_{2}^{(2,1)}$&$0$          &$0$          &$d_{1}^{(1,1)}$&$d_{2}^{(1,1)}$&$0$          &$0$           \\
		$4$&$0$          &$0$          &$d_{1}^{(2,2)}$&$d_{2}^{(2,2)}$&$d_{1}^{(2,1)}$&$d_{2}^{(2,1)}$&$d_{1}^{(1,1)}$&$d_{2}^{(1,1)}$&$d_{1}^{(1,1)}$&$d_{2}^{(1,1)}$&$b'_{1}-b_{1}$&$-b'_{1}-b_{1}$\\
		$5$&$d_{1}^{(1,1)}$&$d_{2}^{(1,1)}$&$d_{1}^{(2,2)}$&$d_{2}^{(2,2)}$&$d_{1}^{(2,1)}$&$d_{2}^{(2,1)}$&$d_{1}^{(1,2)}$&$d_{2}^{(1,2)}$&$d_{1}^{(1,1)}$&$d_{2}^{(1,1)}$&$b'_{2}-b_{2}$&$-b'_{2}-b_{2}$\\
		\hline	\end{tabular}\label{table-2}
\end{table}

\begin{proposition}\label{prop-3}
For $q=6$ and $m\geq 3$, the above Cases (1)-(4) identify $(3700+20m)(m+1)(m-2)$ compatible offsets other than the generalized cases I-V.
\end{proposition}

\begin{proof}
We can verify this enumeration from  the non-zero columns of coefficient matrices in Table \ref{table-2}. Denote the set of coefficient matrices of the offsets in case ($i$) by $\{\mathcal{S}_{i}\}$ from different choices of $\omega_{k}$, $\upsilon_{k}$ and ${\underline{d}^{(k,p_{k})}}$ ($i=1,2,3,4$).

First of all, every coefficient matrix in $\{\mathcal{S}_{i}\}$ $(i=1,2,3,4)$  have  at least 3 non-zero columns $\vec{c}_j$ ($1\leq j\leq m$).
So these coefficient matrices must be different from those in the generalized cases I-V.

Secondly, we prove that $\{\mathcal{S}_{i}\}$  $(i=1,2,3,4)$ are pairwise disjoint. From the positions of the non-zero entries of coefficient matrices in Table \ref{table-2}, it is obviously  $\{\mathcal{S}_{1}\}\bigcap \left(\{\mathcal{S}_{2}\}\bigcup\{\mathcal{S}_{3}\}\bigcup\{\mathcal{S}_{4}\}\right) =\varnothing$. From the definition of the NSGIP, we have both $(b_1'-b_1, b_2'-b_2)\neq (0, 0)$ and $(-b_1'-b_1, -b_2'-b_2)\neq (0, 0)$. Together with the positions of the non-zero entries of $x_{\omega}$ and $x_{\omega+1}$, we obatin $\{\mathcal{S}_{3}\}\bigcap \{\mathcal{S}_{4}\}=\varnothing$.  If one coefficient  matrix of offset belongs to both  $\{\mathcal{S}_{2}\}$ and  $\{\mathcal{S}_{3}\}$ (or $\{\mathcal{S}_{4}\}$), we have $\omega_{2}'=\upsilon_{1}$ and $\omega_{2}'+1=\upsilon_{2}$ in Table \ref{table-2}, which contradicts to $\upsilon_{1}+2\le \upsilon_{2}$ in Cases (3) and (4). Thus we obtain $\{\mathcal{S}_{2}\}\bigcap \left(\{\mathcal{S}_{3}\}\bigcup\{\mathcal{S}_{4}\}\right) =\varnothing$.

Thirdly, it is straightforward that different parameters in each case lead to different offset.

With the same arguments in Proposition \ref{prop-2}, we can prove that there are a total of $(m+1)(m-2)$ ordered pairs $(\omega_{1},\omega_{2})$ in Case (1). For each ordered pair, there are $10$ choices of ${\underline{d}^{(1,1)}}$ such that ${\underline{d}^{(1,1)}}\in \mathcal{C}_{1}$, and there are $(14^2-2\times2^2)$ choices of ${\underline{d}^{(2,1)}}$ and ${\underline{d}^{(2,2)}}$ such that ${\underline{d}^{(2,1)}}, {\underline{d}^{(2,2)}}\notin \mathcal{C}_{4}, ({\underline{d}^{(2,1)}}, {\underline{d}^{(2,2)}})\notin \left(\mathcal{C}_{2},\mathcal{C}_{2}\right)\bigcup\left(\mathcal{C}_{3},\mathcal{C}_{3}\right)$. Thus we have  $\#\{\mathcal{S}_{1}\}=1880\times(m+1)(m-2)$. Similar to Case (1), we also have $\#\{\mathcal{S}_{2}\}=1880\times(m+1)(m-2)$. We consider Cases (3) and (4) together. It was proved in \cite{Liu2013New} that different choices of the subscript $\upsilon_{1}, \upsilon_{2},\upsilon$ and $\vec{b},\vec{b}'$ for NSGIPs over $\mathcal{Q}_{3}$ are  $2\times(m-2)(m+1)$. Moreover, we have $10$
choices of ${\underline{d}^{(1,1)}}$ such that ${\underline{d}^{(1,1)}}\in \mathcal{C}_{1}$, and $(m-3)$ choices of $\omega$ satisfying the conditions in Cases (3) and (4). Thus we have $\#\{\mathcal{S}_{3}, \mathcal{S}_{4}\}=20\times(m-3)(m+1)(m-2)$.

From the discussion above, we obtain
$$\#\{\mathcal{S}_{1}, \mathcal{S}_{2},\mathcal{S}_{3}, \mathcal{S}_{4}\}=\#\{\mathcal{S}_{1}\}+\#\{\mathcal{S}_{2}\}+\#\{\mathcal{S}_{3}\}+\#\{\mathcal{S}_{4}\}=(3700+20m)(m^2-m-2),$$
which completes the proof.
\end{proof}

\section{GAPs and PU Matrices over QAM}

In this section, we show a new viewpoint to construct GCPs over $4^q$-QAM by arrays and paraunitary matrices, by extending the results for PSK case \cite{CCA}.

\subsection{GCPs and GAPs over QAM}

We generalize the concept of the Golay array pair (GAP) from PSK \cite{Array1, Array2} to $4^q$-QAM in this subsection.

An $m$-dimensional complex-valued array of size $\underbrace{2\times2\times\cdots\times2}_m$
can be expressed  by a function $F(x_1, x_2, \cdots, x_{m})$ (or ${F}(\bm{x})$ for short) from $\Z_{2}^{m}$ to $\C$.

\begin{definition}
	The {\em aperiodic auto-correlation} of  an array $F(\bm{x})$ of size $2\times 2 \times \cdots \times 2$ at shift $\bm\tau=(\tau_1, \tau_2, \cdots \tau_{m})$ ($\tau_k=-1, 0 \ \mbox{or} \  1$) is defined by
	\begin{equation}
	C_{{F}}(\bm{\tau})=
	\sum_{\bm{x}}{F(\bm{x}+\bm{\tau})\cdot \overline{F}(\bm{x})},
	\end{equation}
	where \lq\lq$\bm{y}+\bm{\tau}$\rq\rq  is the element-wise addition of vectors over $\Z$, and ${F(\bm{x}+\bm{\tau})\cdot \overline{F}(\bm{x})}=0$ if  $F(\bm{x}+\bm{\tau})$ or $F(\bm{x})$ is not defined.
\end{definition}

\begin{definition}
	A pair of arrays $ \{F(\bm{x}), G(\bm{x})\} $ of size $2\times 2 \times \cdots \times 2$ is said to be a  {\em Golay array pair} (GAP) if
	\begin{equation}\label{GAP1}
	{C}_{F}(\bm{\tau})+{C}_{G}(\bm{\tau})=0, \forall \bm{\tau}\ne\bm{0}.
	\end{equation}
\end{definition}

For further results on GAPs, see \cite{Array2, CCA}.
An array over QPSK of size $2\times 2 \times \cdots \times 2$ can be described by a GBF over $\Z_4$ \cite{CCA}. Similarly,
an $m$-dimensional  array of size $2\times2\times\cdots\times2$ over $4^q$-QAM  can be described by a V-GBF  $\vec{f}(\bm{x})=(f^{(0)}(\bm{x}),f^{(1)}(\bm{x}),\cdots,f^{(q-1)}(\bm{x})): \mathbb{F}_2^m \rightarrow\mathbb{Z}_{4}^q$ by the weighted sum:
\begin{equation} \label{equation_QAMBoolean*}
F(\bm{x})=\sum_{p=0}^{q-1}2^{q-1-p}\cdot  \xi^{f^{(p)}(\bm{x})}.
\end{equation}
Arrays over QPSK can be obviously regarded as arrays over $4^q$-QAM for $q=1$.

A sequence  $F(y)$ of length $2^m$ can be connected with  an array ${F}(\bm{x})$ by setting $y=\sum_{j=1}^{m}x_j\cdot 2^{j-1}$. The aperiodic auto-correlation $C_{F}(\tau)$ of sequence $F(y)$  can be derived from the sum of aperiodic auto-correlation $C_{{F}}(\bm{\tau})$ of array ${F}(\bm{x})$ by restricting $\tau=\sum_{j=1}^{m}2^{j-1}\tau_{j}$, i.e.,
\begin{equation*}
C_{{F}}({\tau})=
\sum_{\bm{\tau}\in\mathcal{D}(\tau)}C_{F}(\bm{\tau}),
\end{equation*}
where $\mathcal{D}(\tau)=\{\bm{\tau}|\tau=\sum_{j=1}^{m}2^{j-1}\tau_{j}\}$.
Thus, if the arrays $F(\bm{x})$ and $G(\bm{x})$ over QAM described by V-GBFs $\vec{f}(\bm{x})$ and $\vec{g}(\bm{x})$ form a GAP, the sequences associated with V-GBFs $\vec{f}(\bm{x})$ and $\vec{g}(\bm{x})$ must form a GCP.

Moreover,  we can construct a large number of GCPs over QAM from a single GAP over QAM by the following theorem, the proof of which is omitted since it is similar to the PSK case in \cite[Lemma 8]{Array2} or \cite[Property 2]{CCA}.
\begin{theorem}\label{them3}
	Suppose that a pair of arrays over QAM described by V-GBFs $\{\vec{f}(\bm{x}),\vec{g}(\bm{x})\}$ form a GAP. Then
	the arrays described by V-GBFs
	$$\left\{\pi\cdot\vec{f}(\bm{x})+{f}'(\bm{x})\cdot\vec{1}, \  \pi\cdot\vec{g}(\bm{x})+{f}'(\bm{x})\cdot\vec{1}\right\}$$
	is also a GAP over QAM, where $\pi\cdot f(\bm{x})=f(x_{\pi(1)},x_{\pi(2)},\cdots, x_{\pi(m)})$ is an arbitrary permutation action on V-GBFs, and
	${f}'(\bm{x})=\sum_{j=1}^{m}c_jx_j+c_{0}$ ($c_j\in \Z_4$) is an arbitrary affine GBF from $\Z_2^m$ to $\Z_4$.
	Consequently, the sequences associated with the above V-GBFs form a GCP.
\end{theorem}

\subsection{GAPs and PU Matrices over QAM}

In this subsection, we extend the theory on GAPs and para-unitary matrices from PSK case \cite{CCA} to QAM case.

The {\em generating function} of  a complex-valued array ${F}(\bm{x})$ of size $2\times 2 \times \cdots \times 2$ is defined by
\begin{equation}\label{arr-gene}
F(z_1, z_2, \cdots, z_{m})=\sum_{x_1,x_2,\cdots, x_{m}}{F(x_1, x_2, \cdots, x_{m})}z_1^{x_1}z_2^{x_2}\cdots z_{m}^{x_{m}},
\end{equation}
(or denoted by
$F(\bm{z})=\sum_{\bm{x}}{F(\bm{x})}\cdot\bm{z}^{\bm{x}}$ for short). It is easy to verify
\begin{equation}
F(\bm{z})\cdot \overline{F}(\bm{z}^{-1})=\sum_{\bm{\tau}}C_{{F}}(\bm{\tau})z_1^{\tau_1}z_2^{\tau_2}\cdots z_{m}^{\tau_{m}},
\end{equation}
where $\bm{z}^{-1}=(z_1^{-1}, z_2^{-1}, \cdots, z_{m}^{-1})$. So arrays $ \{{F}(\bm{x}), {G}(\bm{x})\} $  form a GAP  if and only if their  generating functions $\{F(\bm{z}), G(\bm{z})\}$ satisfy
\begin{equation}\label{GAP2}
F(\bm{z})\cdot \overline{F}(\bm{z}^{-1})+G(\bm{z})\cdot \overline{G}(\bm{z}^{-1})=c,
\end{equation}
where $c$ is a real constant.

Note that the array described by $\vec{f}({\bm{x}})$ over QAM (or ${f}({\bm{x}})$ over QPSK)
is uniquely determined by the generating function ${F}({\bm{z}})$, and vice versa.

We define three types of matrices. Let ${F}_{u,v}(\bm{x})$ $(u,v\in \{0, 1\})$ be $m$-dimensional arrays over $4^q$-QAM  corresponding to V-GBF $\vec{f}_{u,v}(\bm{x})=(f_{u,v}^{(0)}(\bm{x}),f_{u,v}^{(1)}(\bm{x}),\cdots,f_{u,v}^{(q-1)}(\bm{x}))$ and generating function  ${F}_{u,v}(\bm{z})$. These arrays can be expressed by a formalized array matrix $\mathbb{{M}}({\bm{x}})$ with entry ${F}_{u,v}(\bm{x})$, i.e.,
\begin{equation}\label{QAM M(y)}
\mathbb{{M}}({\bm{x}})=
\begin{bmatrix}
{F}_{0,0}({\bm{x}})&{F}_{0,1}({\bm{x}})\\
{F}_{1,0}({\bm{x}})&{F}_{1,1}({\bm{x}})\\
\end{bmatrix}.
\end{equation}
Also, these arrays can be described by a formalized matrix with the V-GBF entry, i.e.,
\begin{equation}\label{QAM M(x)}
\widetilde{{\mathbb{{M}}}}({\bm{x}})=	\begin{bmatrix}
\vec{f}_{0,0}({\bm{x}})&\vec{f}_{0,1}({\bm{x}})\\
\vec{f}_{1,0}({\bm{x}})&\vec{f}_{1,1}({\bm{x}})\\
\end{bmatrix},
\end{equation}
and described by a  matrix with the generating-function entry, i,e.,
\begin{equation}\label{QAM M(z)}
\mathbb{{M}}({\bm{z}})=
\begin{bmatrix}
{F}_{0,0}({\bm{z}})&{F}_{0,1}({\bm{z}})\\
{F}_{1,0}({\bm{z}})&{F}_{1,1}({\bm{z}})\\
\end{bmatrix}.
\end{equation}
$\mathbb{{{M}}}({\bm{z}})$ is called the  generating-function matrix of $\mathbb{{M}}({\bm{x}})$ and $\mathbb{{\widetilde{{M}}}}({\bm{x}})$. For the case $q=1$ (QPSK), these matrices $\mathbb{{M}}({\bm{x}})$, $ \mathbb{{\widetilde{{M}}}}({\bm{x}})$ and $\mathbb{{{M}}}({\bm{z}})$ are denoted by
$\bm{{M}}({\bm{x}})$, $ \bm{{\widetilde{{M}}}}({\bm{x}})$ and $\bm{{{M}}}({\bm{z}})$, respectively, in the following paper.

Similar to the PSK case  \cite[Theorem 1]{CCA}, it is straightforward to obtain the following result.
\begin{theorem}\label{them4}
	Let $\mathbb{{M}}({\bm{z}})$ be the  generating-function matrix of a V-GBF matrix $\widetilde{{\mathbb{{M}}}}({\bm{x}})$. If $\mathbb{{M}}({\bm{z}})$ is a {\em para-unitary} (PU) matrix, i.e.,
	\begin{equation}\label{PU}
	\mathbb{{M}}({\bm{z}})\cdot\mathbb{{M}}^{\dagger}(\bm{z}^{-1})=c \cdot \bm{I},
	\end{equation}
	where  $c$ is a real number, $(\cdot)^\dagger$ denotes the Hermitian transpose and $\bm{I}$ is an identity matrix of order 2, $\mathbb{{M}}({\bm{z}})$ is called a {\em desired PU} matrix.
	The arrays over QAM described by every row (or column) of $\widetilde{{\mathbb{{M}}}}({\bm{x}})$ form a GAP.
\end{theorem}

From Theorem \ref{them4}, GAPs  can be constructed by studying  $\mathbb{{M}}({\bm{z}})$ over QAM satisfying the PU condition.

\section{Constructions of  PU Matrices over $4^q$-QAM}

Before we show our idea on how to construct PU matrices over QAM, we first revisit the  construction of PU matrices $\bm{{{M}}}({\bm{z}})$ over QPSK in \cite{CCA}.
Let  $\bm{H}_k$ $(0\leq k\leq m)$ be arbitrary {\em Butson-type} (BH) Hadamard matrices \cite{Butson62} of order $2$ (with entries being fourth roots of unity), and $\bm{D}(z)=\begin{bmatrix}1&0\\
0& z \end{bmatrix}$. All the  standard GCPs over QPSK can be derived from the following PU matrices over QPSK:
\begin{equation}\label{QPSK-PU}
\bm{M}(\bm{z})=\bm{H}_0\cdot\bm{D}({z}_{1})\cdot\bm{H}_1\cdot\bm{D}(z_2)\cdots\bm{H}_{m-1}\cdot\bm{D}({z}_{m})\cdot\bm{H}_{m}.
\end{equation}

We go back to the matrices over QAM introduced in the previous section. Denote the component GBF matrix of the V-GBF matrix $\widetilde{{\mathbb{{M}}}}({\bm{x}})$ (in the form (\ref{QAM M(x)})) by
\begin{equation*}
\widetilde{\bm{M}}^{(p)}({\bm{x}})=	\begin{bmatrix}
{f}_{0,0}^{(p)}({\bm{x}})&{f}_{0,1}^{(p)}({\bm{x}})\\
{f}_{1,0}^{(p)}({\bm{x}})&{f}_{1,1}^{(p)}({\bm{x}})\\
\end{bmatrix}
\end{equation*}
for $0\le p<q$. Let $\bm{{{M}}}^{(p)}({\bm{z}})$ be the generating-function matrix of $\widetilde{\bm{M}}^{(p)}({\bm{x}})$. Since an array over $4^q$-QAM can be represented by a weighted sum of arrays over QPSK, the generating-function matrix $\mathbb{{M}}({\bm{z}})$ can also be represented by a weighted sum of $\bm{{{M}}}^{(p)}({\bm{z}})$, i.e.,
\begin{equation}\label{QAM PSK M(z)}
\mathbb{{M}}({\bm{z}})=
\sum_{p=0}^{q-1}2^{q-1-p}\bm{{{M}}}^{(p)}({\bm{z}}).
\end{equation}
If $\mathbb{{M}}({\bm{z}})$, {\em the weighted sum of $\bm{{{M}}}^{(p)}({\bm{z}})$, is a PU matrix}, GAPs over QAM are constructed by Theorem \ref{them4}. Moreover, if {\em each $\bm{{{M}}}^{(p)}({\bm{z}})$ is a PU matrix over QPSK}, GBF matrix  $\widetilde{\bm{M}}^{(p)}({\bm{x}})$ can be derived by the method proposed in \cite{CCA}.

However, it is very difficult to choose PU matrices over QSPK with the form (\ref{QPSK-PU}) such that their weighted sum is still a PU matrix. To address this problem, we first study the most simple PU matrices over QPSK, i.e, BH matrices, which can be  uniquely expressed by
\begin{equation}
\bm{{H}}({d_{0}}, {d_{1}}, {d_{2}})= \xi^{d_{0}}\cdot\begin{bmatrix}
1&{0}\\
{0}& \xi^{d_{1}}
\end{bmatrix}\begin{bmatrix}
1&1\\
1&-1
\end{bmatrix}\begin{bmatrix}
1&0\\
0& \xi^{d_{2}}
\end{bmatrix}=\begin{bmatrix}
\xi^{d_{0}}&\xi^{d_{0}+d_{2}}\\
\xi^{d_{0}+d_{1}}&-\xi^{d_{0}+d_{1}+d_{2}}
\end{bmatrix},
\end{equation}
where $d_{0}, d_{1}, d_{2}\in \Z_4$. Let $\mathcal{BH}_0$ be a subset of these BH matrices such that $2d_0+d_1+d_2=0$ over $\Z_4$.
In particular, we denote $\bm{H}(0,0,0)$ by $\bm{H}$ in the following paper, i.e.,
\begin{equation*}
\bm{H}=\begin{bmatrix}
1&1\\
1&-1
\end{bmatrix}
\end{equation*}

\begin{lemma}\label{lemma_QAM H}
	Let ${\bm{H}}_{p} \in \mathcal{BH}_0$ be BH matrices and $c_p$  real numbers for $0\leq p<q$. Then
	$\mathbb{{H}}=\sum_{p=0}^{q-1} c_p\cdot{\bm{H}}_{p}$
	is a unitary  matrix, i.e.,
	$\mathbb{\bm{H}}\mathbb{\bm{H}}^{\dag}=c\cdot\bm{I},$
	where $c$ is a real number.
\end{lemma}
\begin{proof}
	Since $2d_{0}+d_{1}+d_{2}=0$ over $\Z_4$, we have BH matrices
	\begin{equation*}
	\bm{{H}}({d_{0}}, {d_{1}}, {d_{2}})=\begin{bmatrix}
	\xi^{d_{0}}&\xi^{d_{0}+d_{2}}\\
	\xi^{d_{0}+d_{1}}&-\xi^{d_{0}+d_{1}+d_{2}}
	\end{bmatrix}=\begin{bmatrix}
	\xi^{d_{0}}&\overline{\xi^{d_{0}+d_{1}}}\\
	\xi^{d_{0}+d_{1}}&-\overline{\xi^{d_{0}}}
	\end{bmatrix}.
	\end{equation*}
	Then there exist complex number $\alpha$ and $\beta$ such that
	\begin{equation*}
	\mathbb{{H}}=\begin{bmatrix}
	\alpha&\overline{\beta}\\
	\beta&-\overline{\alpha}
	\end{bmatrix},
	\end{equation*}
	It's easy to verify that $\mathbb{\bm{H}}\cdot\mathbb{\bm{H}}^{\dag}=(|\alpha|^{2}+|\beta|^{2})\cdot\bm{I}.$
\end{proof}

We give examples for $q=6$ and $m=2$ to show our idea to construct PU matrices over QAM.
\begin{example}
	Suppose that all the following BH matrices $\bm{H}_{p}, \bm{H}_{p,i} \in \mathcal{BH}_0$. We have
	$$\mathbb{{M}}_{\omega_1=1}({\bm{z}})=\bm{H}_{0}\cdot\bm{D}({z}_{1})\cdot\left(2^0\bm{H}_{1, 0}+2^1\bm{H}_{1, 1}+2^2\bm{H}_{1, 2}+2^3\bm{H}_{1, 3}+2^4\bm{H}_{1, 4}+2^5\bm{H}_{1, 5}+2^6\bm{H}_{1, 6}\right)\cdot\bm{D}(z_2)\cdot\bm{H}_{2}$$
	is a PU matrix by Lemma \ref{lemma_QAM H}. Moreover,
	$$\mathbb{{M}}_{\omega_1=1}({\bm{z}})=\sum_{p=0}^{5}2^{p}\left(\bm{H}_{0}\cdot\bm{D}({z}_{1})\cdot\bm{H}_{1, p}\cdot\bm{D}(z_2)\cdot\bm{H}_{2}\right)$$
	is a weighted sum of PU Matrices over QPSK, so $\mathbb{{M}}_{\omega_1=1}({\bm{z}})$ is a PU matrix over $4^6$-QAM.
	
	We can also construct PU matrices over $4^6$-QAM based on the factorization $6=3\times2$. For example:
	$$\mathbb{{M}}_{(\omega_1, \omega_2)=(1, 0)}({\bm{z}})=\left((2^3)^1\bm{H}_{0,0}+(2^3)^0\bm{H}_{0,1}\right)\cdot\bm{D}({z}_{1})\cdot\left(2^2\bm{H}_{1, 0}+2^1\bm{H}_{1, 1}+2^0\bm{H}_{1, 2}\right)\cdot\bm{D}(z_2)\cdot\bm{H}_{2}.$$
	Moreover, we have
	$$\mathbb{{M}}_{(\omega_1, \omega_2)=(1, 0)}({\bm{z}})=\sum_{p_2=0}^{1}\sum_{p_1=0}^{2}2^{3(1-p_2)+(2-p_1)}\left(\bm{H}_{0, p_2}\cdot\bm{D}({z}_{1})\cdot\bm{H}_{1, p_1}\cdot\bm{D}(z_2)\cdot\bm{H}_{2}\right),$$
	where we can see the mixed radix representation plays an important role.
\end{example}

From the above example, we can construct  PU matrices over QAM by replacing BH matrices with the weighted sums of BH matrices in  $\mathcal{BH}_0$ (satisfying $2d_0+d_1+d_2=0$) in formula (\ref{QPSK-PU}). The \lq weights\rq \ can be determined by the factorization of $q$, and the positions $\{\omega_k\}$ of these weighted sums of BH matrices can be arbitrarily chosen. General results on the constructions of PU matrices over QAM will be introduced in the next sections.

\subsection{First Construction of PU Matrices over QAM}

In this subsection, recall the  factorization of $q=q_1\times q_2\times\cdots\times q_t$, the mappings $\rho_k$, and the vectors ${\underline{d}^{(k,p_{k})}}=\left({d_{0}^{(k,p_{k})}}, {d_{1}^{(k,p_{k})}}, {d_{2}^{(k,p_{k})}}\right)$ ($1\le k\le t$, $0\le p_{k}\le q_{k}{-}1$) in Theorem  \ref{thm_main_1}.

\begin{definition}\label{def-Hp}
	From the vectors ${\underline{d}^{(k,p_{k})}}$, we can define BH matrices $\bm{H}_{k,p_{k}}\in \mathcal{BH}_0$ by
	\begin{equation}
	\bm{H}_{k,p_{k}}={\bm{{H}}}({d_{0}^{(k,p_{k})}}, {d_{1}^{(k,p_{k})}}, {d_{2}^{(k,p_{k})}}).
	\end{equation}
	In particular,  we always have $\bm{H}_{k,0}=\bm{H}$, since $\underline{d}^{(k,0)}=\left(0, 0, 0\right)$. For $1\le k\le t$,  we define  the weighted sums of BH matrices by
	\[\mathbb{H}^{\{k\}}=\sum_{p_k=0}^{q_k-1}2^{(q_{k}-1-p_{k})\cdot\prod_{i=1}^{k-1}q_{i}}\cdot{\bm{H}}_{k,p_{k}}.\]
\end{definition}

\begin{theorem}\label{them5}
	With the notations above, for arbitrary ordered position set $\{\omega_{1},\omega_{2},\dots,\omega_{t}\}\subset \{0,1,\dots,m\}$,
	\begin{equation}\label{pu-1}
	{\mathbb{M}}_{1}(\bm{z})=\bm{U}^{\{0\}}\cdot\prod_{j=1}^{m}\left(\bm{D}({z}_{j})\cdot\bm{U}^{\{j\}}\right)
	\end{equation}
	is a PU matrix over $4^q$-QAM, where
	\begin{equation}
	\bm{U}^{\{j\}}=\left\{
	\begin{aligned}
	&\mathbb{H}^{\{k\}}, & \exists k,\, j=\omega_{k};\\
	&\bm{H},  & \forall l,\, k\neq\omega_{l}.
	\end{aligned}\right.
	\end{equation}
\end{theorem}

\begin{proof}
	According to the mixed radix representation and the definition of $\mathbb{H}^{\{k\}}$, we have
	\begin{equation}\label{QAM-PU-QPSK-1}
	{\mathbb{M}}_{1}(\bm{z})=\sum_{p=0}^{q-1}2^{q-1-p}\cdot\bm{{{M}}}_{1}^{(p)}({\bm{z}})
	\end{equation}
	where $\bm{{{M}}}_{1}^{(p)}({\bm{z}})$ are PU matrices over QPSK with the form
	\begin{equation}\label{PU-QPSK-1}
	\bm{M}_{1}^{(p)}({\bm{z}})
	=\bm{U}_{p,0}\cdot\prod_{j=1}^{m}\left(\bm{D}({z}_{j})\cdot\bm{U}_{p,j}\right),
	\end{equation}
	where
	\begin{equation}
	\bm{U}_{p,j}=\left\{
	\begin{aligned}
	&{\bm{H}}_{k, \rho_k(p)}, & \exists k,\, j=\omega_{k};\\
	&\bm{H},  & \forall k,\, j\neq\omega_{k}.
	\end{aligned}\right.
	\end{equation}
	Then ${\mathbb{M}}_{1}(\bm{z})$ is a PU matrix over $4^q$-QAM by Lemma  \ref{lemma_QAM H}.
\end{proof}

We will give the corresponding GBF matrices of the generating-function matrices $\bm{M}_{1}^{(p)}({\bm{z}})$ in formula (\ref{QAM-PU-QPSK-1}), from which we can prove Theorem  \ref{thm_main_1}.

\subsection{Second Construction of PU Matrices over QAM}

Recall the  factorization of $q=q_0\times q_1\times\cdots\times q_t$, the mappings $\rho'_k$ ($0\le k\le t$), and the vectors ${\underline{d}^{(k,p_{k})}}=\left({d_{0}^{(k,p_{k})}}, {d_{1}^{(k,p_{k})}}, {d_{2}^{(k,p_{k})}}\right)$ ($1\le k\le t$, $0\le p_{k}\le q_{k}{-}1$) given in  Theorem  \ref{thm_main_2}. We now introduce another construction of PU matrices over $4^q$-QAM involving NSGIP.

For NSGIP $Q_0= Q(b_1,b_2,\dots,b_{q_0-1})$ and $Q_1= Q(b'_1,b'_2,\dots,b'_{q_0-1})$, define two matrices
$$diag\{Q_{0},Q_{1}\}=\begin{bmatrix}Q_{0}&0\\
0& Q_{1} \end{bmatrix}\ \text{and} \
\mathbb{Q}=\begin{bmatrix}
Q_{0}&Q_{1}\\
\overline{Q}_{1}&\overline{Q}_{0}
\end{bmatrix}.$$ Then $diag\{Q_{0},Q_{1}\}$ is a unitary matrix. Moreover, if
$\bm{M}(\bm{z})=
\begin{bmatrix}
{F}_{0,0}({\bm{z}})&{F}_{0,1}({\bm{z}})\\
{F}_{1,0}({\bm{z}})&{F}_{1,1}({\bm{z}})\\
\end{bmatrix}$ is a PU matrix, then
\begin{equation*}
\mathbb{Q}\odot\bm{M}(\bm{z}) =
\begin{bmatrix}
Q_{0}\cdot{F}_{0,0}({\bm{z}})&Q_{1}\cdot{F}_{0,1}({\bm{z}})\\
\overline{Q}_{1}\cdot{F}_{1,0}({\bm{z}})&\overline{Q}_{0}\cdot{F}_{1,1}({\bm{z}})\\
\end{bmatrix}
\end{equation*}
is also a PU matrix, where the symbol $\odot$ means the element-wise product of matrices.

\begin{theorem}\label{them6}
	With the notations above, For $1\le k\le t$, define  the weighted sums of BH matrices by
	\[\mathbb{H}^{\{k\}}=\sum_{p_k=0}^{q_k-1}2^{(q_{k}-1-p_{k})\cdot\prod_{i=0}^{k-1}q_{i}}\cdot{\bm{H}}^{(k,p_{k})},\]
	where $\bm{H}^{(k,p_{k})}$ is given in Definition \ref{def-Hp}. For arbitrary ordered position set $\{\omega_{1},\omega_{2},\dots,\omega_{t}\}\subset \{0,1,\dots,m\}$,  $2\le \upsilon\le m-1$, $1\le \upsilon_{1}\le m-2$, $\upsilon_{1}+2\le \upsilon_{2}\le m$,  both matrices
	\begin{equation}\label{pu-2a}
	{\mathbb{M}}_{a}(\bm{z})=\bm{U}^{\{0\}}\cdot \prod_{j=1}^{\upsilon-1}\left(\bm{D}(z_j)\cdot
	{\bm{U}}^{\{j\}}\right)\cdot diag\{Q_{0},Q_{1}\}\cdot\prod_{j=\upsilon}^{m}\left(\bm{D}(z_j)\cdot
	{\bm{U}}^{\{j\}}\right),
	\end{equation}
	and
	\begin{equation}\label{pu-2b}
	{\mathbb{M}}_{b}(\bm{z})=\prod_{j=1}^{\upsilon_{1}}\left(\bm{U}^{\{j-1\}}\cdot\bm{D}(z_{j})\right)\cdot\left(\mathbb{Q}\odot\left({\bm{U}}^{\{\upsilon_{1}\}}\cdot \prod_{j=\upsilon_{1}+1}^{\upsilon_{2}-1}\left(\bm{D}(z_k)\cdot
	{\bm{U}}^{\{j\}}\right)\right)\right)
	\cdot  \prod_{j={\upsilon_{2}}}^{m}\left(\bm{D}(z_k)\cdot
	{\bm{U}}^{\{j\}}\right),
	\end{equation}
	where
	\begin{equation}
	\bm{U}^{\{j\}}=\left\{
	\begin{aligned}
	&\mathbb{H}^{\{k\}}, & \exists k,\, j=\omega_{k};\\
	&\bm{H},  & \forall k,\, j\neq\omega_{k}.
	\end{aligned}\right.
	\end{equation}
	are PU matrices over $4^q$-QAM.
\end{theorem}

\begin{proof}
For $e=a$ and $e=b$, it is easy to check that ${\mathbb{M}}_{e}(\bm{z})$ are PU matrices.
	According to the mixed radix representation and the definition of $\rho'_k$,  we have
	\begin{equation}\label{QAM-PU-QPSK-2}
	{\mathbb{M}}_{e}(\bm{z})=\sum_{p=0}^{q-1}2^{q-1-p}\cdot\bm{{{M}}}_{e}^{(p)}({\bm{z}}),
	\end{equation}
	where $\bm{M}_{e}^{(p)}(\bm{z})$ are multivariate polynomial matrices given by
	\begin{equation}\label{pu-2}
	{\bm{M}_{a}^{(p)}}(\bm{z})=\bm{U}_{p, 0}\cdot \prod_{j=1}^{\upsilon-1}\left(\bm{D}(z_j)\cdot
	{\bm{U}}_{p, j}\right)\cdot diag\{\xi^{b_{{\rho}'_{0}(p)}},\xi^{b'_{{\rho}'_{0}(p)}}\}\cdot\prod_{j=\upsilon}^{m}\left(\bm{D}(z_j)\cdot
	{\bm{U}}_{p,j}\right),
	\end{equation}\\
	and
	\begin{equation}\label{pu-3}
	\begin{split}
	{\bm{M}_{b}^{(p)}}(\bm{z})&=\prod_{j=0}^{\upsilon_{1}-1}\left(\bm{U}_{p, j}\cdot\bm{D}(z_{j+1})\right)\cdot
	\left(\begin{bmatrix}
	\xi^{b_{{\rho}'_{0}(p)}}&  \xi^{b'_{{\rho}'_{0}(p)}}\\
	\xi^{-b'_{{\rho}'_{0}(p)}}& \xi^{-b_{{\rho}'_{0}(p)}}
	\end{bmatrix}\odot\left({\bm{U}}_{p, \upsilon_{1}}\cdot \prod_{j=\upsilon_{1}+1}^{\upsilon_{2}-1}\left(\bm{D}(z_j)\cdot
	{\bm{U}}_{p, j}\right)\right)\right)\\
	&\cdot  \prod_{j=\upsilon_{2}}^{m}\left(\bm{D}(z_j)\cdot
	{\bm{U}}_{p, j}\right),
	\end{split}
	\end{equation}
	where
	\begin{equation}
	\bm{U}_{p,j}=\left\{
	\begin{aligned}
	&{\bm{H}}_{k, \rho'_k(p)}, & \exists k,\, j=\omega_{k};\\
	&\bm{H},  & \forall k,\, j\neq\omega_{k}.
	\end{aligned}\right.
	\end{equation}
By checking both ${\bm{M}_{a}^{(p)}}(\bm{z})$ and ${\bm{M}_{b}^{(p)}}(\bm{z})$ are PU matrices over QPSK, we finish the proof.
\end{proof}

We will give the corresponding GBF matrices of the generating-function matrices $\bm{M}_{a}^{(p)}({\bm{z}})$ and $\bm{M}_{b}^{(p)}({\bm{z}})$ in formulae (\ref{pu-2}) and (\ref{pu-3}), respectively, from which we can prove Theorem  \ref{thm_main_2}.

\section{Extracting Corresponding V-GBFs}

In this section, we will  develop a method to extract the corresponding V-GBF  matrices from their generating matrices ${\mathbb{M}}_{1}(\bm{z})$,  ${\mathbb{M}}_{a}(\bm{z})$ and ${\mathbb{M}}_{b}(\bm{z})$ introduced in Theorems \ref{them5} and \ref{them6}. We prove Theorems \ref{thm_main_1} and \ref{thm_main_2} by showing that all the GCPs in Theorems \ref{thm_main_1} and \ref{thm_main_2} can be respectively obtained by the corresponding V-GBF matrices of PU matrices ${\mathbb{M}}_{1}(\bm{z}), {\mathbb{M}}_{a}(\bm{z})$ and ${\mathbb{M}}_{b}(\bm{z})$, according to   Theorems \ref{them3} and \ref{them4}.

\subsection{GBF Matrices and Their Generating Matrices}

In this subsection, we introduce some basic results on how to extract GBF matrices from their generating matrices over QPSK.

The following notations of matrices of order $2$ will be used in the rest of the paper.

\begin{itemize}
	\item $\bm{D}(z)=\begin{bmatrix}1&0\\
	0& z \end{bmatrix}$,
	$\bm{D}(x)=\begin{bmatrix}1{-}x&0\\
	0& x \end{bmatrix}.$
	\item $\bm{H}=\begin{bmatrix}
	1&1\\
	1&-1
	\end{bmatrix},
	\widetilde{\bm{H}}=\begin{bmatrix}
	0&0\\
	0&2
	\end{bmatrix}.$
	\item $
	\bm{J}=\begin{bmatrix}1&1\\
	1&1 \end{bmatrix},
	\bm{A}=\begin{bmatrix}0&0\\
	1&1 \end{bmatrix},
	\bm{B}=\begin{bmatrix}0&1\\
	0&1 \end{bmatrix}.$
\end{itemize}

In the following theorem, suppose that
$\{\bm{z}_0,\bm{z}_1,\dots,\bm{z}_{m},{z}_1,\dots,z_{m}\}$ are  multivariate variables which do not intersect with each other, and  $\{\bm{x}_0,\bm{x}_1,\dots,\bm{x}_{m},{x}_1,\dots,x_{m}\} $ are their corresponding Boolean variables respectively, where $\bm{z}_k$ and $\bm{x}_k$
are multivariate variables, and  ${z}_k$ and ${x}_k$ are single variables.
\begin{theorem}\label{Boolean_matrix}
	For $0\le j\leq m$, let $\bm{M}^{\{j\}}(\bm{z}_j)$ be generating matrices of GBF matrices $\widetilde{\bm{M}}^{\{j\}}(\bm{x}_j)$ over QPSK.
	Denote $\bm{z}=(\bm{z}_0,\bm{z}_1,\dots,\bm{z}_{m},{z}_1,\dots,z_{m})$ and
	$\bm{x}=(\bm{x}_0,\bm{x}_1,\dots,\bm{x}_{m},{x}_1,\dots,x_{m})$. Then the corresponding  GBF matrix
	of
	\begin{equation}\label{QPSK-generating-matrix}
	{\bm{M}}(\bm{z})=\bm{M}^{\{0\}}(\bm{z}_0)\cdot\left(\prod_{j=1}^{m}\left(\bm{D}(z_j)\cdot \bm{M}^{\{j\}}(\bm{z}_j)\right)\right)
	\end{equation} is given by
	\begin{equation}\label{QPSK-GBF-matrix}
	\bm{\widetilde{M}}(\bm{x})
	=\bm{\widetilde{M}}^{\{0\}}(\bm{x}_0)\cdot\bm{D}(x_{1})\cdot\bm{J}+\sum_{j=1}^{m-1}\bm{J}\cdot
	\bm{D}(x_{j})\cdot\bm{\widetilde{M}}^{\{j\}}(\bm{x}_j)\cdot\bm{D}(x_{j+1})\cdot\bm{J}+\bm{J}\cdot
	\bm{D}(x_{m})\cdot\bm{\widetilde{M}}^{\{m\}}(\bm{x}_m).
	\end{equation}
\end{theorem}
\begin{proof}
	See Appendix A.
\end{proof}

The following corollary is an immediate consequence of Theorem \ref{Boolean_matrix}. Note that it has been provided in both \cite{Budi2018PU} and \cite{CCA} by other expressions, but the matrix expression shown here will simplify the process  to extract the V-GBF from PU matrices given in  Theorems 5 and 6.

\begin{corollary}\label{coro}
	The corresponding GBF matrix of the following PU matrix
	\begin{equation}\label{equation_standard_U(z)}
	\bm{U}(\bm{z})=\bm{H}\cdot\left(\prod_{j=1}^{m}\left(\bm{D}(z_j)\cdot \bm{H}\right)\right)
	\end{equation}
	is given by
	\begin{equation}\label{Uorder2}
	\widetilde{\bm{U}}(\bm{x})=f(\bm{x})\cdot\bm{J}+2x_{1}\cdot\bm{A}+2x_{m}\cdot\bm{B},
	\end{equation}
	where ${f}(\bm{x})={2}\cdot\sum_{j=1}^{m-1}x_{j}x_{j+1}$.
\end{corollary}
\begin{proof}
	For $0\le j\leq m$, let $\bm{M}^{\{j\}}(\bm{z}_j)=\bm{H}$ in Theorem \ref{Boolean_matrix}. By verifying the following results,
	\begin{eqnarray*}
		\bm{\widetilde{H}}\cdot\bm{D}(x_{1})\cdot\bm{J}&=&2x_{1}\cdot\bm{A},\\
		\bm{J}\cdot
		\bm{D}(x_{j})\cdot\bm{\widetilde{H}}\cdot\bm{D}(x_{j+1})\cdot\bm{J}&=&2x_jx_{j+1}\cdot \bm{J},\\
		\bm{J}\cdot
		\bm{D}(x_{m})\cdot\bm{\widetilde{H}}&=&2x_{m}\cdot\bm{B},
	\end{eqnarray*}
	the proof is completed by formula (\ref{QPSK-GBF-matrix}).
\end{proof}

\subsection{Extracting Corresponding V-GBFs from PU Matrices}

Recall the PU matrices over QAM: ${\mathbb{M}}_{1}(\bm{z})$, ${\mathbb{M}}_{a}(\bm{z})$ and ${\mathbb{M}}_{b}(\bm{z})$,  constructed in Theorems \ref{them5} and \ref{them6}.
We have shown that they are all weighted sums of the PU matrices $\bm{{{M}}}_{e}^{(p)}({\bm{z}})$ ($e=1, a$ or $b$) over QPSK. Moreover, all PU matrices $\bm{{{M}}}_{e}^{(p)}({\bm{z}})$ ($e=1, a$ or $b$) over QPSK are with  the form (\ref{QPSK-generating-matrix}), and a method to extract GBF matrices from these PU matrices has been developed in Theorem \ref{Boolean_matrix}, from which we are able to extract V-GBF matrices $\widetilde{\mathbb{M}}_{e}(\bm{x})$ (for $e=1, a$ or $b$).

{\begin{theorem}\label{them8}
		Let $\vec{d}$ vectors $\vec{d_{i}}^{(k)}$ and  $\vec{d_{i}}^{[k]}$, $\vec{b}$ vectors and  position set $\{\omega_{1},\omega_{2},\dots,\omega_{t}, \upsilon, \upsilon_{1}, \upsilon_{2}\}$ be given in Theorems \ref{thm_main_1} and \ref{thm_main_2} for $e=1, a$ or $b$. Then
		${\mathbb{M}}_{e}(\bm{z})$ is the generating matrix of the
		V-GBF matrix
		\begin{equation}\label{eqn_thm5}
		\widetilde{\mathbb{M}}_{e}(\bm{x})
		=\left(f(\bm{x})\cdot\vec{1}+\vec{s}(\bm{x})\right)\cdot\bm{J}
		+\vec{\mu}_{A}(\bm{x})\cdot\bm{A}
		+\vec{\mu}_{B}(\bm{x})\cdot\bm{B},
		\end{equation}	
		where $f(\bm{x})={2}\cdot\sum_{j=1}^{m-1}x_{j}x_{j+1}$.
		
		For the case $e=1$,
		
		$$\vec{s}(\bm{x})=\sum_{k=1}^{t}\left(\vec{d_{1}}^{(k)}x_{\omega_{k}}+\vec{d_{2}}^{(k)}x_{{\omega_{k}}+1}+\vec{d_{0}}^{(k)}\right),$$
		
		$$\vec{\mu}_{A}(\bm{x})=\left\{
		\begin{aligned}
		&2x_{1}\cdot\vec{1}+\vec{d_{1}}^{(k)},&\exists k,\,\omega_{k}=0,\\
		&2x_{1}\cdot\vec{1},&\text{otherwise};
		\end{aligned}\right.$$

		$$\vec{\mu}_{B}(\bm{x})=\left\{
		\begin{aligned}
		&2x_{m}\cdot\vec{1}+\vec{d_{2}}^{(k)},&\exists k,\,\omega_{k}=m;\\
		&2x_{m}\cdot\vec{1},&\text{otherwise}.
		\end{aligned}\right.$$

		For the case $e=a$,
		$$\vec{s}(\bm{x})=\sum_{k=1}^{t}\left(\vec{d_{1}}^{[k]}x_{\omega_{k}}+\vec{d_{2}}^{[k]}x_{{\omega_{k}}+1}+\vec{d_{0}}^{[k]}\right)+(\vec{b}'-\vec{b})x_{\upsilon}+\vec{b}.$$

		For the case $e=b$,
		$$\vec{s}(\bm{x})=\sum_{k=1}^{t}\left(\vec{d_{1}}^{[k]}x_{\omega_{k}}+\vec{d_{2}}^{[k]}x_{{\omega_{k}}+1}+\vec{d_{0}}^{[k]}\right)+
		(\vec{b}'-\vec{b})x_{\upsilon_{1}}+(-\vec{b}'-\vec{b})x_{\upsilon_{2}}+\vec{b}.$$
		
		For both cases $e=a$ and $e=b$,
		$$\vec{\mu}_{A}(\bm{x})=\left\{
		\begin{aligned}
		&2x_{1}\cdot\vec{1}+\vec{d_{1}}^{[k]},&\exists k,\,\omega_{k}=0,\\
		&2x_{1}\cdot\vec{1},&\text{otherwise}.
		\end{aligned}\right.$$
		$$\vec{\mu}_{B}(\bm{x})=\left\{
		\begin{aligned}
		&2x_{m}\cdot\vec{1}+\vec{d_{2}}^{[k]},&\exists k,\,\omega_{k}=m,\\
		&2x_{m}\cdot\vec{1},&\text{otherwise}.
		\end{aligned}\right.$$
	\end{theorem}
	
	\begin{proof}
		See Appendix B.
	\end{proof}
	
	Notice that the V-GBFs in first row and first column of  $\widetilde{\mathbb{M}}_{e}(\bm{x})$ in Theorem \ref{them8} are
	\begin{equation*}
	\left\{
	\begin{aligned}
	&\vec{f}(\bm{x})=f(\bm{x})\cdot\vec{1}+\vec{s}(\bm{x}),\\
	&\vec{g}(\bm{x})=\vec{f}(\bm{x})+\vec{\mu}_{A}(\bm{x}),
	\end{aligned}\right.
	\quad\text{and}\quad
	\left\{
	\begin{aligned}
	&\vec{f}(\bm{x})=f(\bm{x})\cdot\vec{1}+\vec{s}(\bm{x}),\\
	&\vec{g}(\bm{x})=\vec{f}(\bm{x})+\vec{\mu}_{B}(\bm{x}),
	\end{aligned}\right.
	\end{equation*}
	respectively, which are both GAPs by Theorem \ref{them4}. Then by applying Theorem \ref{them3}, the results in Theorems \ref{thm_main_1} and \ref{thm_main_2} follow immediately.

\section{Concluding Remarks}

In this paper, we propose two new constructions of GCSs over $4^q$-QAM of length $2^{m}$, which are usually presented by the combination of the standard GCSs over QPSK and compatible offsets. We propose new compatible offsets by introducing the so-called $\vec{d}$-vectors based on the factorization of integer $q$ from the Set $\mathcal{C}$ and $\vec{b}$-vectors from NSGIPs as the ingredients. We show that the generalized cases I-V \cite{Li2010A,Liu2013New} are special cases of our new constructions. Moreover, our constructions significantly increased the number of the GCSs over $4^q$-QAM if $q$ is a  composite number. It has been shown in \cite{Li2010A} and \cite{Liu2013New} that the numbers of offsets in the generalized cases I-III and IV-V are a linear  polynomial of $m$ and a  quadratic  polynomial of $m$, respectively. We show that, for $q=q_{1}\times q_{2}\times \dots\times q_{t}$ ($q_k>1$), the number of new offsets in our first construction is lower bounded by a polynomial of $m$ with degree $t$.  In particular, for $q=4$, the numbers of  new  offsets in our first construction is seven times more than that in the generalized  cases IV-V. We also show that the numbers of  new  offsets in our two constructions is lower bounded by a cubic polynomial of $m$ for $q=6$.

The results are proved from new viewpoints of GAPs and PU matrices over QAM. 

We show  that a large number of GCPs can be constructed from a single GAP over $4^q$-QAM in Theorem \ref{them3}, which extend the idea in \cite{Array2} for PSK  case to  QAM.
This argument greatly simplifies the process for constructing GCPs and GCSs over $4^q$-QAM. Although the  three-stage process in  \cite{Array2} are not involved here, Theorem \ref{them3}  has the same importance as the three-stage process for   GAPs of size $2\times 2 \times \cdots \times 2$.   Our proof implies that all the mentioned  GCSs over QAM in this paper can be regarded as projections of Golay complementary arrays of size $2\times2\times\cdots\times2$, so the results in this paper provide a partial solution to an open
problem from \cite{Array2} for GAPs of size $2\times 2 \times \cdots \times 2$:

{\em How can the three-stage construction process  be used to simplify or extend
known results on the construction of Golay sequences in QAM modulation? }

A full answer to this question for GAPs of size $L_1\times L_2 \times \cdots \times L_m$ over QAM will be given in our future work.

We also make a connection between GAPs and specified PU matrices with multi-variables over QAM in Theorem \ref{them4}, which generalizes the idea in  \cite{Budi2018PU} for GCPs and  PU matrices with a single variable over QAM. The PU matrices constructed in Theorems \ref{them5} and \ref{them6} can be easily decomposed to the weighted sum of PU 
matrices over QPSK, which make sure that we can derive the GBF form of corresponding GAPs. It should be pointed out that the GCSs proposed here belong to the so-called $M$-Qum cases, which was mentioned,  but the GBF form could not be explicitly given in \cite{Budi2018PU}. On the other hand, many new GCSs over QAM, realized by the PU algorithm, were also found in \cite{Budi2018PU} by  exhaustive search. For instance, compared with the generalized cases I-V,  numerical results showed that the overall increase in the total number of GCSs  including 1-Qum and 2-Qum cases of length 1024  is up to $59\%, 242\%,$ and $340\%$, for 64-, 256-, and 1024-QAM, respectively. We left the new constructions of the PU matrices over QAM and the corresponding V-GBFs, which can significantly increase the number of GCSs proposed in this paper,  as an open problem.

\section*{Acknowledgment}

The authors wish to thank Dr. S.~Z.~Budi\v{s}in for his valuated discussions on the constructions of PU in \cite{Budi2018PU} and the constructions in this paper.

\appendices
\section{Proof of Theorem \ref{Boolean_matrix}}
We claim that
\begin{equation}
{\bm{M}}(\bm{z}_0,\bm{z}_1,z)=\bm{M}^{\{0\}}(\bm{z}_0)\cdot\bm{D}(z)\cdot\bm{M}^{\{1\}}(\bm{z}_{1})
\end{equation}
is the generating matrix of  GBF matrix
\begin{eqnarray}\label{eqn_proof}
\bm{\widetilde{M}}(\bm{x}_0,\bm{x}_1,x)
=\bm{\widetilde{M}}^{\{0\}}(\bm{x}_0)\cdot\bm{D}(x)\cdot\bm{J}+\bm{J}\cdot
\bm{D}(x)\cdot\bm{\widetilde{M}}^{\{1\}}(\bm{x}_1).
\end{eqnarray}
over QPSK. To prove the claim,
recall that $\bm{D}(z)=\sum_{x=0}^{1}\bm{D}(x)\cdot z^{x}$. Then we have
\begin{eqnarray*}
	{\bm{M}}(\bm{z}_0,\bm{z}_1,z)&=&
	\bm{M}^{\{0\}}(\bm{z}_0)\cdot\bm{D}(z)\cdot\bm{M}^{\{1\}}(\bm{z}_{1})\\ &=&\sum_{\bm{x}_0}\left({\bm{M}^{\{0\}}(\bm{x}_0)}\cdot\bm{z}_{0}^{\bm{x}_0}\right)\cdot
	\sum_{x}\left(\bm{D}(x)\cdot z^{x}\right)\cdot \sum_{\bm{x}_1}\left({\bm{M}^{\{1\}}(\bm{x}_1)}\cdot\bm{z}_{1}^{\bm{x}_1}\right)\\
	&=&\sum_{\bm{x}_0}\sum_{x}\sum_{\bm{x}_1}{\bm{M}^{\{0\}}(\bm{x}_0)}\cdot\bm{D}(x)\cdot {\bm{M}^{\{1\}}(\bm{x}_1)}\cdot\bm{z}_{0}^{\bm{x}_0}\cdot z^{x}\cdot\bm{z}_{1}^{\bm{x}_1}.
\end{eqnarray*}
On the other hand, from another expansion:
$${\bm{M}}(\bm{z}_0,\bm{z}_1,z)=\sum_{\bm{x}_0}\sum_{x}\sum_{\bm{x}_1}{\bm{M}(\bm{x}_0,\bm{x}_1,x)}\cdot\bm{z}_{0}^{\bm{x}_0}\cdot z^{x}\cdot\bm{z}_{1}^{\bm{x}_1},$$  we obtain
\begin{equation*}
{\bm{M}(\bm{x}_0,\bm{x}_1,x)}={\bm{M}^{\{0\}}(\bm{x}_0)}\cdot\bm{D}(x)\cdot {\bm{M}^{\{1\}}(\bm{x}_1)}.
\end{equation*}
Then the entry of ${\bm{M}(\bm{x}_0,\bm{x}_1,x)}$ can be expressed by
\begin{equation*}
M_{i,j}(\bm{x}_0,\bm{x}_1,x)=M^{\{0\}}_{i,x}(\bm{x}_{0})\cdot M^{\{1\}}_{x,j}(\bm{x}_{1}).
\end{equation*}
Alternatively, we have
\begin{equation}\label{eqn2-proof}
\widetilde{M}_{i,j}(\bm{x}_0,\bm{x}_1,x)=\widetilde{M}^{\{0\}}_{i,x}(\bm{x}_{0})+ \widetilde{M}^{\{1\}}_{x,j}(\bm{x}_{1})
\end{equation}
for $i,j, x=0, 1$. Then we can verify that the formulae (\ref{eqn_proof}) and (\ref{eqn2-proof}) are equivalent by specifying $x=0$ and $1$.

Furthermore, Theorem \ref{Boolean_matrix} can be proved by iteratively using the above claim.

\section{Proof of Theorem \ref{them8}}

Denote GBF matrix of $ \bm{H}(d_{0}, d_{1}, d_{2}) $ by $\widetilde{\bm{{H}}}(d_{0}, d_{1}, d_{2})$, we have
\begin{equation}\label{eqn_H(ddd)}
\widetilde{\bm{{H}}}({d_{0}}, {d_{1}}, {d_{2}})
=\begin{bmatrix}
{d_{0}}&{d_{0}+d_{2}}\\
{d_{0}+d_{1}}&{d_{0}+d_{1}+d_{2}}+2
\end{bmatrix}
={d_{0}}\cdot\bm{J}+{d_{1}}\cdot\bm{A}+{d_{2}}\cdot\bm{B}+\begin{bmatrix}
0&0\\
0&2
\end{bmatrix}.
\end{equation}

\begin{lemma}\label{lem2}
	Let the matrices $\bm{D}(x)$, $\bm{J}$, $\bm{A}$ and $\bm{B}$ be the same at those given in Subsection 7.1, we have
	\begin{itemize}
		\item[(1)] $\bm{J}\cdot\bm{D}(x)\cdot\bm{J}=\bm{J}$;
		\item[(2)] $\bm{A}\cdot\bm{D}(x)\cdot\bm{J}=\bm{A}$;
		\item[(3)] $\bm{B}\cdot\bm{D}(x)\cdot\bm{J}=x\cdot\bm{J}$;
		\item[(4)] $\bm{J}\cdot\bm{D}(x)\cdot\bm{A}=x\cdot\bm{J}$;
		\item[(5)] $\bm{J}\cdot\bm{D}(x)\cdot\bm{B}=\bm{B}$.
	\end{itemize}
\end{lemma}

Now we can give the proof of Theorem \ref{them8} for the case $e=1$.

\begin{proof}( Extracting V-GBF from ${\mathbb{M}}_{1}(\bm{z})$ ) Recall the notations in the proof of Theorem \ref{them5}. As $\rho_{k}(0)=0$ ($1\le k\le t$),  we have  ${\bm{H}}^{(k,\rho_{k}(0))}=\bm{H}$, which leads to  $\bm{{{M}}}_{1}^{(0)}({\bm{z}})=\bm{H}\cdot\left(\prod_{j=1}^{m}\left(\bm{D}(z_j)\cdot \bm{H}\right)\right)=\bm{U}({\bm{z}})$ and  $\bm{\widetilde{M}}_{1}^{(0)}(\bm{x})=\widetilde{\bm{U}}(\bm{x})$, shown in Corollary \ref{coro}.

 According to  Theorem \ref{Boolean_matrix}, we can obtain the GBF matrices $\bm{\widetilde{M}}_{1}^{(p)}(\bm{x})$ of PU matrices $\bm{{{M}}}_{1}^{(p)}({\bm{z}})$.
We study the difference between  $\bm{\widetilde{M}}_{1}^{(p)}(\bm{x})$ and $\bm{\widetilde{M}}_{1}^{(0)}(\bm{x})$:
	\begin{equation*}
	\begin{split}
	\bm{\widetilde{M}}_{1}^{(p)}(\bm{x})-\bm{\widetilde{M}}_{1}^{(0)}(\bm{x})
	&=\left(\sum_{k=1}^{t}\bm{J}\cdot
	\bm{D}(x_{\omega_{k}})\cdot(\widetilde{\bm{H}}^{(k,{\rho}_k(p))}-\widetilde{\bm{H}})\cdot\bm{D}(x_{\omega_{k}+1})\cdot\bm{J}\right)\\
	&+(\bm{I}-\bm{J}\cdot\bm{D}(x_{0}))\cdot(\widetilde{\bm{U}}_{p,0}-\widetilde{\bm{H}})\cdot\bm{D}(x_{1})\cdot\bm{J}\\
	&+\bm{J}\cdot
	\bm{D}(x_{m})\cdot(\widetilde{\bm{U}}_{p, m}-\widetilde{\bm{H}})\cdot(\bm{I}-\bm{D}(x_{m+1})\cdot\bm{J}),
	\end{split}
	\end{equation*}
	where where  $x_{0}$ and $x_{m+1}$ are \lq fake\rq \  variables.

	According to (\ref{eqn_H(ddd)}), the difference of $\widetilde{\bm{H}}^{(k,p_{k})}$ and $\widetilde{\bm{H}}$ can be expressed by
	$$(\bm{\widetilde{H}}^{(k,p_{k})}-\bm{\widetilde{H}})= {d_{0}^{(k,p_{k})}}\cdot\bm{J}+{d_{1}^{(k,p_{k})}}\cdot\bm{A}+{d_{2}^{(k,p_{k})}}\cdot\bm{B}.$$

	From Lemma \ref{lem2}, each term in the sum of the above difference between  $\bm{\widetilde{M}}_{1}^{(p)}(\bm{x})$ and $\bm{\widetilde{M}}_{1}^{(0)}(\bm{x})$ can be respectively simplified as
	\begin{align*}
	\sum_{k=1}^{t}\bm{J}\cdot
	\bm{D}(x_{\omega_{k}})\cdot(\widetilde{\bm{H}}^{(k,{\rho}_k(p))}-\widetilde{\bm{H}})\cdot\bm{D}(x_{\omega_{k}+1})\cdot\bm{J}
	&=\sum_{k=1}^{t}\left(d_{1}^{(k,{\rho}_k(p))}x_{\omega_{k}}+d_{2}^{(k,{\rho}_k(p))}x_{\omega_{k}+1}+d_{0}^{(k,{\rho}_k(p))}\right)
	\cdot\bm{J},
	\\
	(\bm{I}-\bm{J}\cdot\bm{D}(x_{0}))\cdot(\widetilde{\bm{U}}_{p,0}-\widetilde{\bm{H}})\cdot\bm{D}(x_{1})\cdot\bm{J}
	&=\left\{
	\begin{aligned}
	&d_{1}^{(k,{\rho}_{k}(p))}\cdot\bm{A},&\exists k,\,\omega_{k}=0,\\
	&\bm{0},&\text{otherwise};
	\end{aligned}\right.
	\\
	\bm{J}\cdot
	\bm{D}(x_{m})\cdot(\widetilde{\bm{U}}_{p, m}-\widetilde{\bm{H}})\cdot(\bm{I}-\bm{D}(x_{m+1})\cdot\bm{J})
	&=\left\{
	\begin{aligned}
	&d_{2}^{(k,{\rho}_{k}(p))}\cdot\bm{B},&\exists k,\,\omega_{k}=m,\\
	&\bm{0},&\text{otherwise}.
	\end{aligned}\right.
	\end{align*}
	By applying Corollary \ref{coro}, we obtain
	\begin{equation}\label{final-1}
	\bm{\widetilde{M}}_{1}^{(p)}(\bm{x})=
	\left(f(\bm{x})+{s}^{(p)}(\bm{x})\right)\cdot \bm{J}+\mu_{A}^{(p)}\cdot\bm{A}+\mu_{B}^{(p)}\cdot\bm{B},
	\end{equation}
	where  $f(\bm{x})={2}\cdot\sum_{j=1}^{m-1}x_{j}x_{j+1}$, and
	\begin{itemize}
		\item
		${s}^{(p)}(\bm{x})=\sum_{k=1}^{t}\left(d_{1}^{(k,{\rho}_{k}(p))}x_{\omega_{k}}+d_{2}^{(k,{\rho}_{k}(p))}x_{{\omega_{k}}+1}+d_{0}^{(k,{\rho}_{k}(p))}\right)$,
		where   $x_{0}=x_{m+1}=0$;
		\item
		$
		{\mu}_{A}^{(p)}(\bm{x})=\left\{
		\begin{aligned}
		&2x_{1}+d_{1}^{(k,{\rho}_{k}(p))},&\exists k,\,\omega_{k}=0,\\
		&2x_{1},&\text{otherwise};
		\end{aligned}\right.$
		\item $
		{\mu}_{B}^{(p)}(\bm{x})=\left\{
		\begin{aligned}
		&2x_{m}+d_{2}^{(k,{\rho}_{k}(p))},&\exists k,\omega_{k}=m,\\
		&2x_{m},&\text{otherwise}.
		\end{aligned}\right.
		$
	\end{itemize}
	This completes the proof.
\end{proof}

\begin{lemma}\label{lem4}
	Let $\bm{M}(\bm{z})$ be the generating matrix of GBF matrix $\widetilde{\bm{M}}(\bm{x})$ over QPSK. Suppose $\alpha, \beta, c_{i,j}(0\le i,j\le 1)\in \Z_{4}$. We have
	\begin{itemize}
		\item[(1)]$\bm{M}(\bm{z})\cdot diag\{\xi^{\alpha},\xi^{\beta}\}$ is the generating matrix of GBF matrix $\widetilde{\bm{M}}(\bm{x})+\bm{J}\cdot diag\{{\alpha},{\beta}\}$;
		\item[(2)]$\bm{M}(\bm{z})\odot\bm{C}$ (or $\bm{C}\odot\bm{M}(\bm{z})$)is the generating matrix of GBF matrix $\widetilde{\bm{M}}(\bm{x})+\widetilde{\bm{C}}$, where
		$\bm{C}=\begin{bmatrix}\xi^{c_{0,0}}&\xi^{c_{0,1}}\\
		\xi^{c_{1,0}}& \xi^{c_{1,1}} \end{bmatrix}$ and
		$\widetilde{\bm{C}}=\begin{bmatrix}
		c_{0,0}&c_{0,1}\\
		c_{1,0}& c_{1,1} \end{bmatrix}.$
	\end{itemize}
\end{lemma}

Now we can give the proof of Theorem \ref{them8} for the cases $e=a$ and $e=b$.

\begin{proof}(( Extracting V-GBF from ${\mathbb{M}}_{a}(\bm{z})$ and ${\mathbb{M}}_{b}(\bm{z})$) Recall the notations in the proof of Theorem \ref{them6}.
	According to Lemma \ref{lem4}, by iteratively using Theorem \ref{Boolean_matrix}, the  corresponding GBF matrices are given by
	\begin{equation*} \bm{\widetilde{M}}_{a}^{(p)}(\bm{x})=\bm{\widetilde{M}}_1^{(p)}(\bm{x})+\bm{J}\cdot\bm{D}(x_{\upsilon-1})\cdot\bm{J}\cdot diag\{{b_{{\rho}'_{0}(p)}},{b'_{{\rho}'_{0}(p)}}\}\cdot\bm{D}(x_{\upsilon})\cdot\bm{J}
	\end{equation*}
	and
	\begin{equation*}
	\bm{\widetilde{M}}_{b}^{(p)}(\bm{x})=
	\bm{\widetilde{M}}_1^{(p)}(\bm{x})+\bm{J}\cdot\bm{D}(x_{\upsilon_{1}})\begin{bmatrix}
	{b_{{\rho}'_{0}(p)}}& {b'_{{\rho}'_{0}(p)}}\\
	{-b'_{{\rho}'_{0}(p)}}&{-b_{{\rho}'_{0}(p)}}
	\end{bmatrix}\cdot\bm{D}(x_{\upsilon_{2}})\cdot\bm{J},
	\end{equation*}
	where
	\begin{equation}
	\bm{\widetilde{M}}_1^{(p)}(\bm{x})
	=\bm{\widetilde{U}}_{p,0}\cdot\bm{D}(x_{1})\cdot\bm{J}+\sum_{j=1}^{m-1}\bm{J}\cdot
	\bm{D}(x_{j})\cdot\bm{\widetilde{U}}_{p,j}\cdot\bm{D}(x_{j+1})\cdot\bm{J}+\bm{J}\cdot
	\bm{D}(x_{m})\cdot\bm{\widetilde{U}}_{p,m}.
	\end{equation}
	The last term in $\bm{\widetilde{M}}_{a}^{(p)}(\bm{x})$ can be calculated by
	\begin{equation}\label{final-2}
	\bm{J}\cdot\bm{D}(x_{\upsilon-1})\cdot\bm{J}\cdot diag\{{b_{{\rho}'_{0}(p)}},{b'_{{\rho}'_{0}(p)}}\}\cdot\bm{D}(x_{\upsilon})\cdot\bm{J}=((b'_{{\rho}'_{0}(p)}-b_{{\rho}'_{0}(p)})x_{\upsilon}+b_{{\rho}'_{0}(p)})\cdot\bm{J}.
	\end{equation}
	The last term in $\bm{\widetilde{M}}_{b}^{(p)}(\bm{x})$ can be calculated by
	\begin{equation}\label{final-3}
	\bm{J}\cdot
	\bm{D}(x_{\upsilon_{1}})\cdot\begin{bmatrix}
	{b_{{\rho}'_{0}(p)}}& {b'_{{\rho}'_{0}(p)}}\\
	{-b'_{{\rho}'_{0}(p)}}&{-b_{{\rho}'_{0}(p)}}
	\end{bmatrix}\cdot\bm{D}(x_{\upsilon_{2}})\cdot\bm{J} =((b'_{{\rho}'_{0}(p)}-b_{{\rho}'_{0}(p)})x_{\upsilon_{1}}+(-b'_{{\rho}'_{0}(p)}-b_{{\rho}'_{0}(p)})x_{\upsilon_{2}}+{b_{{\rho}'_{0}(p)}})\cdot\bm{J}.
	\end{equation}
	The term $\bm{\widetilde{M}}_1^{(p)}(\bm{x})$ is well studied in the proof of  Extracting V-GBF from ${\mathbb{M}}_{1}(\bm{z})$. By replacing the subscript $\rho_{k}(p)$ by $\rho'_{k}(p)$ in formula (\ref{final-1}), we have
	\begin{equation}\label{final-4}
	\bm{\widetilde{M}}_1^{(p)}(\bm{x})=
	\left(f(\bm{x})+{s}_{0}^{(p)}(\bm{x})\right)\cdot \bm{J}+\mu_{A}^{(p)}\cdot\bm{A}+\mu_{B}^{(p)}\cdot\bm{B},
	\end{equation}
	where  $f(\bm{x})={2}\cdot\sum_{j=1}^{m-1}x_{j}x_{j+1}$, $ {s}_{0}^{(p)}(\bm{x})=\sum_{k=1}^{t}\left(d_{1}^{(k,{\rho}'_{k}(p))}x_{\omega_{k}}+d_{2}^{(k,{\rho}'_{k}(p))}x_{{\omega_{k}}+1}+d_{0}^{(k,{\rho}'_{k}(p))}\right) $,
	\[
	{\mu}_{A}^{(p)}(\bm{x})=\left\{
	\begin{aligned}
	&2x_{1}+d_{1}^{({\rho}'_{k}(p))},&\exists  k,\,\omega_{k}=0;\\
	&2x_{1},&\text{otherwise}.
	\end{aligned}\right.
	\quad\text{and}\quad
	{\mu}_{B}^{(p)}(\bm{x})=\left\{
	\begin{aligned}
	&2x_{m}+d_{2}^{({\rho}'_{k}(p))},&\exists k,\,\omega_{k}=m;\\
	&2x_{m},&\text{otherwise}.
	\end{aligned}\right.
	\]
	where  $x_{0}=x_{m+1}=0$ are \lq fake\rq\  variables.	
	
	Combining the formulae  (\ref{final-2}), (\ref{final-3}) and (\ref{final-4}), we complete the proof.
\end{proof}

\section{Proof of Proposition \ref{prop-0}}

For factorization $q=q_{1}\times q_{2}\times \dots\times q_{t}$ ($q_{k}\ge 2$), we have the  mixed radix representation of  $p=(\rho_{t}(p),\rho_{t-1}(p),\dots, \rho_{1}(p))_{q_{t}, q_{t-1}\dots q_{1}}$,  the offsets in Theorem \ref{thm_main_1} are given by
\begin{equation*}
{s}^{(p)}(\bm{x})=\sum_{k=1}^{t}\left(d_{1}^{(k,{\rho}_{k}(p))}x_{\omega_{k}}+d_{2}^{(k,{\rho}_{k}(p))}x_{{\omega_{k}}+1}+d_{0}^{(k,{\rho}_{k}(p))}\right), (0\le p< q),
\end{equation*}
We consider the  vectors ${\underline{d}^{(k,p_{k})}}$  and the ordered sets $\{\omega_{1}, \omega_{2},\dots,\omega_{t}\}$
satisfying the following conditions:
\begin{enumerate}
	\item[(1)]
	for $q_{k}=2$,
	${\underline{d}^{(k,1)}}\in \mathcal{C}_{1}$;
	
	\item[(2)] for $q_{k}\ge 3$, $ {\underline{d}^{(k,p_{k})}}\notin \mathcal{C}_{4}$, and
	$ \left({\underline{d}^{(k,0)}}, {\underline{d}^{(k,1)}},\dots,{\underline{d}^{(k,q_{k}-1)}}\right)\notin \left(\mathcal{C}_{2},\mathcal{C}_{2},\dots,\mathcal{C}_{2}\right)\bigcup\left(\mathcal{C}_{3},\mathcal{C}_{3},\dots,\mathcal{C}_{3}\right)$;
	\item[(3)]
	for any $1\le k_{1}\ne k_{2}
	\le t$,
	$(\omega_{k_1},\omega_{k_2})\neq (0, m)$ or  $(m, 0)$ and $|\omega_{k_1}-\omega_{k_2}|\geq 2 $.
\end{enumerate}

\begin{table}
	\centering
	\caption{Non-zero columns in coefficient matrices of offsets for $q=q_{1}\times q_{2}\times \dots\times q_{t}$}	
	\begin{threeparttable}
		\begin{tabular}{|c|cc@{}ccc@{}cc@{}ccc|}
			\hline
			&\multicolumn{10}{c|}{$ {s}^{(p)}(\bm{x}) $} \\
			\cline{2-11}
			$p_{ }$&$c_{\omega_{1}}$&$c_{\omega_{1}+1}$& &$c_{\omega_{2}}$&$c_{\omega_{2}+1}$& &$\dots$ & &$c_{\omega_{t}}$&$c_{\omega_{t}+1}$\\
			\hline
			$0$&$0$          &$0$       &\vline &$0$          &$0$          &\vline &$ \dots $&\vline &$0$&$0$ \\
			$1$&$d_{1}^{(1,1)}$&$d_{2}^{(1,1)}$&\vline &$0$          &$0$          &\vline &$ \dots $&\vline &$0$&$0$ \\
			\vdots&\vdots&\vdots&\vline &\vdots&\vdots &\vline &\vdots&\vline &\vdots&\vdots \\
			$q_{1}-1$&$d_{1}^{(1,q_{1}-1)}$&$d_{2}^{(1,q_{1}-1)}$&\vline &$0$          &$0$          &\vline &$ \dots $&\vline &$0$&$0$ \\
			\cline{2-3}
			$q_{1}$&$0$          &$0$          &\vline &$d_{1}^{(2,1)}$&$d_{2}^{(2,1)}$&\vline &$ \dots $&\vline &$0$&$0$ \\
			$q_{1}+1$&$d_{1}^{(1,1)}$&$d_{2}^{(1,1)}$&\vline &$d_{1}^{(2,1)}$&$d_{2}^{(2,1)}$&\vline &$ \dots $&\vline &$0$&$0$ \\
			\vdots&\vdots&\vdots&\vline &\vdots&\vdots &\vline &\vdots&\vline &\vdots&\vdots \\
			$q_{1}\cdot2-1$&$d_{1}^{(1,q_{1}-1)}$&$d_{2}^{(1,q_{1}-1)}$&\vline &$d_{1}^{(2,1)}$&$d_{2}^{(2,1)}$&\vline &$ \dots $&\vline &$0$&$0$ \\
			\cline{2-3}
			\vdots&\vdots&\vdots& &\vdots&\vdots &\vline &\vdots&\vline &\vdots&\vdots \\
			\cline{2-3}
			$q_{1}\cdot(q_{2}-1)$&$0$          &$0$          &\vline &$d_{1}^{(2,q_{2}-1)}$&$d_{2}^{(2,q_{2}-1)}$&\vline &$ \dots $&\vline &$0$&$0$ \\
			$q_{1}\cdot(q_{2}-1)+1$&$d_{1}^{(1,1)}$&$d_{2}^{(1,1)}$&\vline &$d_{1}^{(2,q_{2}-1)}$&$d_{2}^{(2,q_{2}-1)}$&\vline &$ \dots $&\vline &$0$&$0$ \\
			\vdots&\vdots&\vdots&\vline &\vdots&\vdots & \vline&\vdots&\vline &\vdots&\vdots \\
			$q_{1}\cdot q_{2}-1$&$d_{1}^{(1,q_{1}-1)}$&$d_{2}^{(1,q_{1}-1)}$&\vline &$d_{1}^{(2,q_{2}-1)}$&$d_{2}^{(2,q_{2}-1)}$&\vline &$ \dots $&\vline &$0$&$0$ \\
			\cline{2-6}
			\vdots&\vdots&\vdots&&\vdots&\vdots&&$ \ddots $ &\vline &\vdots&\vdots \\
			\cline{2-8}
			$\prod_{k=1}^{t-1}q_{k}$&$0$&$0$&&$0$&$0$&&\dots&\vline &$d_{1}^{(t,1)}$&$d_{2}^{(t,1)}$ \\
			\vdots&\vdots&\vdots&&\vdots&\vdots &&\vdots&\vline &\vdots&\vdots \\
			$\prod_{k=1}^{t-1}q_{k}\cdot2-1$&$d_{1}^{(1,q_{1}-1)}$&$d_{2}^{(1,q_{1}-1)}$&&$d_{1}^{(2,q_{2}-1)}$&$d_{2}^{(2,q_{2}-1)}$&&$ \dots $&\vline &$d_{1}^{(t,1)}$&$d_{2}^{(t,1)}$ \\\cline{2-8}
			\vdots&\vdots&\vdots&&\vdots&\vdots &&\vdots&&\vdots&\vdots \\
			\cline{2-8}
			$\prod_{k=1}^{t-1}q_{k}\cdot(q_{t}-1)$&$0$&$0$&&$0$&$0$&&\dots&\vline &$d_{1}^{(t,q_{t}-1)}$&$d_{2}^{(t,q_{t}-1)}$ \\
			\vdots&\vdots&\vdots&&\vdots&\vdots &&\vdots&\vline &\vdots&\vdots \\
			$\prod_{k=1}^{t}q_{k}-1$&$d_{1}^{(1,q_{1}-1)}$&$d_{2}^{(1,q_{1}-1)}$&&$d_{1}^{(2,q_{2}-1)}$&$d_{2}^{(2,q_{2}-1)}$&&$ \dots $&\vline &$d_{1}^{(t,q_{t}-1)}$&$d_{2}^{(t,q_{t}-1)}$ \\
			\hline
		\end{tabular}\label{table-4}
		\begin{tablenotes}
			\item Note:
			The items enclosed by blocks of the same width are periodically identical.
		\end{tablenotes}
	\end{threeparttable}
\end{table}

The condition (3) guarantees that the subscripts of variables $\{x_{\omega_{k}}, x_{\omega_{k}+1}|1\le k\le t\}$ are all different, i.e, there are at most  one
\lq fake\rq \ variable. Then the possible non-zero columns in the corresponding coefficient matrices are given by
\begin{equation}\label{non-zero}
\vec{c}=\sum_{k=1}^t\vec{d_0}^{(k)},\quad
\vec{c}_{\omega_{k}}=\vec{d_{1}}^{(k)},\quad
\vec{c}_{\omega_{k}+1}=\vec{d_{2}}^{(k)},\quad (1\le k\le t),
\end{equation}
which was shown in Table \ref{table-4} in detail. The conditions (1) and (2) guarantee that the $\vec{d}$-vector $\vec{d_{i}}^{(k)}=\left(d_{i}^{(k,{\rho}_{k}(0))},d_{i}^{(k,{\rho}_{k}(1))},\dots,d_{i}^{(k,{\rho}_{k}(q-1))}\right)\ne\vec{0}$ for $i=1,2$ and $1\le k\le t$.
Thus, there are $2t-1$ or $2t$ non-zero columns $\vec{c}_j$ in  coefficient matrices.
Moreover, for $1 \leq k \leq t$ and $1\leq p_k\leq q_k-1$, different choices of $d_1^{(k,p_{k})}$  and  $d_2^{(k,p_{k})}$   lead to different $\vec{d}$-vector  $\vec{d_{i}}^{({k})}$, and  different choices of $d_0^{(k,p_{k})}$  lead to different vector $\vec{c}$ from the  mixed radix representation. Thus, different choices  of $\left({\underline{d}^{(k,0)}}, {\underline{d}^{(k,1)}},\dots,{\underline{d}^{(k,q_{k}-1)}}\right)$  and ordered set $\{\omega_{1}, \omega_{2},\dots,\omega_{t}\}$ satisfying the above conditions determine different coefficient matrices (or different offsets $\vec{s}(\bm{x})$) with at least $2t-1$ non-zero columns $\vec{c}_{\omega_{k}}$.

We first count the number of sets $\{\omega_{1}, \omega_{2},\dots,\omega_{t}\}$ such that $0\leq \omega_{1}<\omega_{2}\cdots <\dots,\omega_{t}\leq m$ satisfying the condition (3). If $\omega_{1}=0$, suppose that the elements in $\{\omega_{2}, \omega_{3},\dots,\omega_{t}\}$ are bars and the other $m-t+1$ elements in the set $\{1, 2, \cdots m\}$ are stars,  a configuration satisfying the condition (3) is obtained by placing $t-1$ separating bars at places between two stars. Since there are $m-t$ gaps between the stars, there are $\tbinom{m-t}{t-1}$ possible configurations. If $\omega_{1}\neq 0$, suppose that the elements in $\{\omega_{1}, \omega_{2},\dots,\omega_{t}\}$ are bars and the other $m-t+2$ elements in the set $\{0, 1, \cdots, m{+}1\}$ are stars,  a configuration satisfying the condition (3) is obtained by placing $t$ separating bars at places between two stars. Since there are $m-t+1$ gaps between the stars, there are $\tbinom{m-t+1}{t}$ possible configurations. Then there are
$$t!\left(\tbinom{m-t}{t-1}+\tbinom{m-t+1}{t}\right)=(m+1)\frac{(m-t)!}{(m-2t+1)!}$$
ordered sets ($\omega_{1}, \omega_{2},\dots,\omega_{t}$) satisfying the condition (3).

For each $k$ satisfying $q_{k}=2$, there are $10$ choices of ${\underline{d}^{(k,1)}}$ satisfying the condition (1). For each $k$ satisfying $q_{k}\ge 3$,
there are $(14^{q_{k}-1}-2\times2^{q_{k}-1})$ choices of ${\underline{d}^{(k,p_{k})}}$  satisfying the condition (2). It is obvious that $(14^{q_{k}-1}-2\times2^{q_{k}-1})=10$ for $q_{k}=2$.

From the discussion above, the conditions (1)(2)(3) identify
$$(m+1)\frac{(m-t)!}{(m-2t+1)!}\prod_{k=1}^{t}(14^{q_{k}-1}-2\times2^{q_{k}-1})$$
compatible offsets with at least $2t-1$ non-zero columns $\vec{c}_j$ in their corresponding coefficient matrices.

\end{document}